\documentclass[aap,MSNbibl,seceqn,citesort,dvips]{arximspdf}

\doi{10.1214/11-AAP802}
\volume{22}
\issue{4}
\pubyear{2012}
\firstpage{1576}
\lastpage{1610}

\makeatletter

\newtheorem{theorem}{Theorem}[section]
\newtheorem{cor}[theorem]{Corollary}
\newtheorem{lem}[theorem]{Lemma}
\newtheorem{prop}[theorem]{Proposition}

\newproclaim{ass}[theorem]{Assumption}

\newproclaim{rem}[theorem]{Remark}
\newproclaim{exa}[theorem]{Example}

\newcommand{\Real}{\mathbb R}
\newcommand{\Natural}{\mathbb N}
\newcommand{\nin}{n \in\Natural}
\newcommand{\eps}{\varepsilon}
\newcommand{\such}{\mid}

\newcommand{\prob}{\mathbb{P}}
\newcommand{\qprob}{\mathbb{Q}}

\newcommand{\F}{\mathcal{F}}
\newcommand{\B}{\mathcal{B}}

\newcommand{\ud}{ d}

\newcommand{\blambda}{\lambda^*}
\newcommand{\bPi}{\Pi^*}
\newcommand{\bprob}{\prob^*}
\newcommand{\breta}{\eta^*}
\newcommand{\bell}{\ell^*}
\newcommand{\bV}{V^*}

\newcommand{\limn}{\lim_{n \to\infty}}
\newcommand{\limt}{\lim_{t \to\infty}}

\newcommand{\plim}{\prob\mbox{-} \lim}
\newcommand{\pliminft}{\prob\mbox{-} \liminft}

\newcommand{\liminft}{\liminf_{t \to\infty}}

\newcommand{\V}{\mathcal{V}}

\newcommand{\Lb}{\mathbb{L}}

\newcommand{\hE}{\hat{E}}

\newcommand{\indic}{\mathbb{I}}

\newcommand{\llloc}{\ll_{\mathrm{loc}}}

\newcommand{\reals}{\mathbb R}

\makeatother

\begin{document}
\begin{frontmatter}

\title{Robust maximization of asymptotic growth}
\runtitle{Robust asymptotic growth}

\begin{aug}
\author[A]{\fnms{Constantinos} \snm{Kardaras}\corref{}\thanksref{t1}\ead[label=e1]{kardaras@bu.edu}}
\and
\author[B]{\fnms{Scott} \snm{Robertson}\ead[label=e2]{scottrob@andrew.cmu.edu}}
\runauthor{C. Kardaras and S. Robertson}
\affiliation{Boston University and Carnegie Mellon University}
\address[A]{Department of Mathematics and Statistics\\
Boston University\\
111 Cummington Street\\
Boston, Massachusetts 02215\\
USA\\
\printead{e1}}
\address[B]{Department of Mathematical Sciences\\
Carnegie Mellon University\\
Wean Hall 6113\\
Pittsburgh, Pennsylvania 15213\\
USA\\
\printead{e2}} 
\end{aug}

\thankstext{t1}{Supported in part by the NSF Grant DMS-09-08461.}

\received{\smonth{6} \syear{2010}}
\revised{\smonth{7} \syear{2011}}

%
\begin{abstract}
This paper addresses the question of how to invest in a robust
growth-optimal way in a market where the instantaneous expected return
of the underlying process is unknown. The optimal investment strategy
is identified using a generalized version of the principal
eigenfunction for an elliptic second-order differential operator, which
depends on the covariance structure of the underlying process used for
investing. The robust growth-optimal strategy can also be seen as a
limit, as the terminal date goes to infinity, of optimal arbitrages in
the terminology of Fernholz and Karatzas [\textit{Ann. Appl. Probab.}
\textbf{20} (2010) 1179--1204].
\end{abstract}

%

\end{frontmatter}

\section*{Discussion} 

This paper addresses the question of how to invest optimally in a
market when the financial planning horizon is long, and the dynamics of
the underlying assets are uncertain. For long time-horizons, it is
reasonable to question whether fixed parameter estimation, especially for
drift rates, remain valid. Therefore, determining a robust way to invest
across potential model misidentifications is desirable, if not
indispensable.

On the canonical space of continuous functions from $[0,\infty)$ to
$\Real^d$, let~$X$ denote
the coordinate mapping, which should be thought as representing the (relative)
price of certain underlying assets, discounted by some baseline wealth
process. It is assumed that there exists a probability $\qprob$ under
which~$X$ has dynamics of the form $\ud X_t = \sigma(X_t)\,\ud W^\qprob_t$, where $c
:=
\sigma\sigma'$ represents the instantaneous covariance matrix, and
$W^\qprob$
is a standard Brownian motion under $\qprob$. The significance of the local
martingale probability $\qprob$ lies in that it acts as a ``dominating''
measure used to form a class of probabilities $\Pi$, out of which an unknown
representative is supposed to capture the true dynamics of the process. The
class $\Pi$ is built by exactly all probabilities satisfying the
following two
conditions:
\begin{itemize}
\item First, under $\prob\in\Pi$ the coordinate mapping~$X$
stays in an open and connected subset $E \subseteq\reals^d$.
Qualitatively, if~$X$ represents either asset prices or relative
capitalizations, this condition asserts that assets
should not cease to exist over the time horizon.
\item Second, for $t\geq0$, each $\prob\in\Pi$ is absolutely continuous
with respect to $\qprob$ on $\sigma(X_s,0\leq s\leq t)$. This last
fact implies that the volatility process of~$X$ under each $\prob\in
\Pi$
is the same; even though model misidentification is possible, the allowable
models are not permitted to be wildly inconsistent with one another.
\end{itemize}
Note that the family $\Pi$ as described above does not necessarily
induce any
ergodic or stability property of the assets, although it certainly contains
all such models; in particular, models where the assets display transient
behavior are allowable. Furthermore, it is \textit{not} assumed that
$\qprob\in\Pi$. Indeed, it is often the case that~$X$ ``explodes'' under
$\qprob$; more precisely, with $\zeta$ denoting the first exit time
of~$X$
from~$E$, $\qprob[\zeta< \infty] > 0$ is allowed.

There are good reasons to let the class of models be defined in the
above way. While the covariance structure given by the function $c$ is
easy to assess, the returns process of~$X$ under the ``true''
probability is statistically impossible to estimate in
practice.\setcounter{footnote}{1}\footnote{Actually, under continuous-time observations,
perfect estimation of $c$ is possible. More realistically,
high-frequency data give good estimators for $c$. In contrast, consider
a one-dimensional model for an asset-price of the form $\ud X_t / X_t =
b \,\ud t + 0.2 \,\ud W_t$, where $b \in\Real$---note that $\sigma= 0.2$ is
considered a ``typical'' value for annualized volatility. Given
observations $(X_t)_{t \in[0, T]}$, where $T > 0$, the best linear
unbiased estimator for $b$ is $\hat{b}_T :=(1 / T) \log(X_T /
X_0)$. Easy calculations show that in order for $|\hat{b}_T - b|
\leq0.01$ to happen with probability at least $95\%$, one needs $T
\approx1600$ (in years). This simple exercise demonstrates the
futility of attempting to estimate drifts.}

Given that the underlying dynamics are only specified within a range
of models $\prob\in\Pi$, a natural question is to find a reasonable criterion
for ``optimal investment in~$X$.'' Here, optimal investment is defined
as a wealth process which ensures the largest possible worst-case (with
respect to the whole class of models) asymptotic growth rate. Given the
set $\V$ of all possible positive stochastic integrals against~$X$
starting from unit initial capital, the asymptotic growth rate of $V
\in\V$ under $\prob\in\Pi$ is defined as the largest~\mbox{$\gamma\in
\Real_+$} such that $\lim_{t\uparrow\infty} \prob[(1/ t)\log
V_t\geq\gamma] = 1$ holds. (An alternative definition of asymptotic
growth rate via almost-sure limits is also considered in the paper.)
With this definition, the investor seeks
to find a wealth process in~$\V$ that achieves maximal growth rate
uniformly over all possible models in $\Pi$, or at least in a~large
enough suitable subclass of $\Pi$ that covers all ``nonpathological''
cases. 

The solution to the above problem is given in terms of a generalized
version of the principal eigenvalue-eigenvector pair $(\blambda,
\breta
)$ of the eigenvalue equation
%
%
\begin{equation}\label{Eeigenvalueproblem}
\frac{1}{2} \sum_{i,j=1}^{d}c_{i,j}(x)\,\frac{\partial^2 \eta
}{\partial
x_i\,\partial x_j} (x) = -\lambda\eta(x), \qquad  x \in E.
\end{equation}
More precisely, the main result of Section~\ref{Sminmaxresult} states
that, when restricted to a large sub-class $\bPi$ of $\Pi$, $\lambda^*$
is the maximal growth rate, and the process $V \in\V$ defined via $V_t
= e^{\lambda^* t} \eta^*(X_t)$ achieves this maximal growth rate. There
are, of course, technicalities on an analytical level arising from the
use of the eigenvalue equation (\ref{Eeigenvalueproblem}), since it is
unreasonable in the present setting to assume either that $c$ is
uniformly positive definite on $E$ or that $E$ is bounded with smooth
boundary. [Consider, e.g., the case where~$X$ represents the prices of $d$
assets. In this instance $E=(0,\infty)^d$, which is unbounded with
corners. Furthermore, once the stock price goes to zero, it remains
stuck there. Thus, the covariance matrix $c$ degenerates along the
boundary of $E$ and hence cannot be both continuous and uniformly
elliptic.] In order to allow for degenerate $c$ and unbounded $E$ with
nonsmooth boundary, but still retain some tractability in the
problem, it is assumed that $E$ can be ``filled up'' by bounded
subregions with smooth boundary and that $c$ is continuous and
pointwise strictly positive definite. Under this assumption,~\cite
{MR1326606}, Chapter~4,
gives a detailed account of eigenvalue equations of the form~(\ref{Eeigenvalueproblem}).

Growth-optimal trading in the face of model uncertainty has been
investigated by other authors. One strand of research considers the
case where asset returns are assumed stationary and ergodic. In
\cite{MR929084}, asymptotically growth-optimal trading strategies
based upon historical data are constructed. There have been a number
of follow-up papers on this topic; see~\cite{MR1159579},
\cite{MR2319423} and the references cited within. In contrast to the
aforementioned approach, knowledge of the entire past is not
required in this paper. In fact, the optimal strategy is only based
on the current level of~$X$ and is, therefore, closely-related to
the idea of functionally-generated portfolios studied in
\cite{MR1861997}. Furthermore, it is also not assumed here that~$X$
represents asset returns; in fact, the primary example is when~$X$
are relative capitalizations, and not asset returns.
In this setting, stationarity of the relative capitalizations does
not automatically transfer to stationarity of returns.

The concept of robust growth optimality is also related to that of
robust utility optimization, the idea of which dates back to
\cite{gilboa1989meu} and is considered in detail in
\cite{MR2211122,MR2247836,MR2096294,MR2284014} and
\cite{MR2236457},  amongst others. (There is also
recent literature on optimal stopping under model ambiguity---see,
e.g.,~\cite{BayKarYao10}.) Though this paper differs from those
mentioned in
not considering penalty functions and by focusing on growth rather than
general utility functions, the growth-optimal strategy provides a ``good''
long-term robust optimal strategy for general utility functions due to the
exponential increase in terminal wealth as time progresses. Two recent papers
which are close in spirit to the present paper are~\cite{Knipsel1} and
\cite{Knipsel2}. Reference~\cite{Knipsel1} considers long-run robust utility
maximization in the case of model uncertainty for power and logarithmic
utility, and~\cite{Knipsel2} addresses the problem of finding wealth processes
that minimize long-term downside risk. The precise manner in which the class
of models is defined in these papers can only be identified up to a
(stochastic) affine perturbation away from a fixed model. This paper
differs from the above two in that, to the extent that
underlying economic factors affect the asset dynamics, it is only
through the
drift of~$X$. Furthermore, there is no a priori fixed model from which all
other models are recovered via perturbations. This enables the class of models
to be determined by qualitative properties, without additional technical
restrictions. However, here, as well as in~\cite{Knipsel2}, there is a
fundamental PDE, playing the role of an ergodic Bellman equation, which
governs the robust trading strategies.

The problem of constructing robust growth-optimal strategies can be extended
to the case where even the covariance matrix $c$ is not known
precisely, but
rather assumed to belong to a class of admissible matrices $\mathcal{C}$.
Such a situation has been studied in~\cite{KarFern10c}, in the setting of
optimal arbitrage mentioned below. In such a setting, one does not even assume
the existence of a dominating probability $\qprob$, and the
probabilities in
$\prob$ can be mutually singular. It is left for future research to establish
a natural definition of an ``extremely'' robust growth-optimal trading
strategy in terms of sub-solutions of (\ref{Eeigenvalueproblem}) which are
uniform over $\mathcal{C}$.

A second goal of the present paper is to relate robust growth-optimal
trading strategies to optimal arbitrages, as considered in
\cite{KarFern10a}. Optimal arbitrages are trading strategies
designed to outperform the benchmark process used for discounting
almost surely over a given time horizon. In~\cite{KarFern10a}, it was
shown that, under certain assumptions, the existence of an optimal
arbitrage on a finite time horizon $[0, T]$, $T \in\Real_+$, is
equivalent to $\qprob[\zeta\leq T] > 0$ (positive probability of
explosion of the coordinate process under $\qprob$ before~$T$), when~$E$
is the simplex in $\reals^d$. In fact, optimal arbitrages are
naturally expressed in terms of (conditional) tails of the distribution
of $\zeta$ under
$\qprob$.

For a fixed $T>0$, denote by $(V_t^T)_{t \in[0, T]}$ the optimal
arbitrage in the interval $[0, T]$. The robust growth-optimal wealth
processes $(V_t)_{t\in\Real+}$ considered here can be regarded as a
long-term limit of the optimal arbitrages; this is a topic taken up in
Section~\ref{Srelarb}. A better understanding of this connection
requires exploring a particular probability $\prob^*$, under which~$X$
has dynamics of the form $\ud X_t = (c(X_t) \nabla \log\breta(X_t)) \,\ud t +
\sigma(X_t)\,\ud W^{\prob^*}_t$ for $t \in \Real _+$, where~$W^{\prob^*}$ is
a standard Brownian motion under $\prob^*$. Loosely speaking,
ergodicity of~$X$ under $\bprob$ implies that on any compact time
interval $[0,\tau]$ the collection of processes
$((V^T_t)_{t\in[0,\tau]})_{T\in\Real+}$ converges to the robust
growth-optimal wealth process $(V_t)_{t\in[0,\tau]}$ as the horizon $T$
becomes large. This is part of the reason why Section~\ref{Sbprob} is
devoted to investigating the properties of~$X$ under $\bprob$. An
application of ergodic results for unbounded functions from \cite
{MR1152459}, coupled with powerful probabilistic arguments, allows us
to show the aforementioned convergence of optimal arbitrages to the
robust growth-optimal one. Furthermore, convergence of the
probabilities $\qprob[ \cdot\such\zeta> T]$ to $\bprob$ on
$\mathcal{F}_\tau$ as $T\uparrow\infty$ in the total-variation norm is
established. This extends results on diffusions conditioned to remain
in a bounded region, first obtained in~\cite{MR781410}, to regions with
nonsmooth boundaries where the matrix $c$ need not be uniformly
positive definite, and where the process~$X$ under $\qprob$ need not be
$m$-reversing for any measure~$m$.

In the special one-dimensional case, considered in Section \ref
{Sone-dim-case}, simple tests for transience and recurrence of
diffusions are readily available. This allows us to provide tight
conditions upon $c$ in the case of a bounded interval, in which
$\blambda= 0$ or $\blambda> 0$, and characterize both the nature of
$\breta$ and of $\bprob$. The main message is essentially the
following: if
$X$ can explode to both endpoints under $\qprob$, then everything
works out nicely, in the sense that $\blambda> 0$ and~$X$ is
positive recurrent under~$\bprob$. The technical proof of this result
relies heavily on singular Sturm--Liouville theory and is given in
Section~\ref{Sonedimproofs}.

Finally, Section~\ref{Sexamples} provides examples that illustrate the
results obtained in previous sections. In contrast to the case where
$c$ is uniformly positive definite on $E$, multi-dimensional examples
where the function $\breta$ does not vanish on the boundary of $E$,
even if $E$ is bounded, are given.

\section{The set-up} 

Consider an open and connected set $E \subseteq\Real^d$ and a function
$c$ mapping $E$ to the space of $d\times d$ matrices. For
$\alpha\in(0,1]$, recall that a~function $f\dvtx E\mapsto\reals$ is
called \textit{locally $C^{2,\alpha}$ on $E$} if for all bounded, open,
connected $D\subset E$ such that $\bar{D}\subset E$ it follows that
$f\in C^{2,\alpha}(\bar{D})$. For a definition of the H\"{o}lder
space~$C^{2,\alpha}$, see~\cite{MR1625845}, Chapter 5.1. The following
assumptions will be in force throughout.
%
%
\begin{ass} \label{assbasic}
For each $x\in E$, $c(x)$ is a symmetric and strictly positive definite
$d \times d$ matrix. For $1\leq i,j\leq
d$, $c_{ij}(x)$ is locally $C^{2,\alpha}$ on $E$ for
some $\alpha\in(0, 1]$. Furthermore, there exists a sequence
$(E_n)_{\nin}$
of bounded open connected subsets of $E$ such that each boundary
$\partial E_n$ is $C^{2, \alpha}$, $\bar{E}_n\subset E_{n+1}$ for
$\nin
$ and $E = \bigcup_{n=1}^{\infty} E_n$.
\end{ass}

\subsection{The generalized martingale problem on $E$} %

It will now be discussed how Assumption~\ref{assbasic} implies
the existence of a unique solution to the generalized martingale
problem on
$E$ for the operator $L$ which acts on $f\in C^2(E)$ via
%
%
\begin{equation}\label{Eoperatordef}
(Lf)(x) = \frac{1}{2}\sum_{i,j=1}^{d} c_{ij}(x)\,\frac{\partial
^2f}{\partial
x_i\,\partial x_j}(x), \qquad  x \in E.
\end{equation}
%

Let $\hE= E \cup\triangle$ be the one-point compactification of
$E$; the point $\triangle$ is identified with $\partial E$ if $E$
is bounded and with $\partial E $ plus the point at $\infty$ if $E$
is unbounded. Let $C ( \Real_+,\hE)$ be the space of
continuous functions from $[0,\infty)$ to $\hE$. For $\omega\in
C(\Real_{+}, \hE)$, define the exit times
\begin{eqnarray*}
\zeta_n(\omega) &:=& \inf\{t \in\Real_+ \such\omega_t
\notin E_n\},\\
\zeta(\omega) &:=& \lim_{n\uparrow\infty} \zeta_n(\omega).
\end{eqnarray*}
Then define
\[
\Omega= \{\omega\in C(\Real_{+},\hE) \such
\omega_{\zeta+ t} = \triangle\mbox{ for all } t \in\Real_+ \mbox{
if } \zeta(\omega) < \infty\}.
\]

Let $X = (X_t)_{t \in\Real_+}$ be the
coordinate mapping process for $\omega\in\Omega$.
Set $\B= (\B_t)_{t \in\Real_+}$ to be the natural filtration
of~$X$. It follows that the smallest $\sigma$-algebra that is generated
by $\bigcup_{t \in\Real_+} \B_t$, denoted by $\B_\infty$, is actually
the Borel $\sigma$-algebra on~$\Omega$. Furthermore, $\B_\infty$ is
also the smallest $\sigma$-algebra that is generated by $\bigcup
_{\nin
}\B_{\zeta_n}$, since paths in $\Omega$ stay in $\triangle$ upon arrival.


A solution to the generalized martingale problem on $E$ is a family
of probability measures $(\qprob_x)_{x\in\hE}$ such
that $\qprob_{x}[X_0 = x] = 1$ and
\[
f(X_{t\wedge\zeta_n}) - \int_0^{t\wedge\zeta_n}(Lf)(X_s)\,ds
\]
is a $(\Omega, (\B_t)_{t \in\Real_+},\qprob_{x})$-martingale
for all $\nin$ and all $f\in C^{2}(E)$ with $Lf$ given as in (\ref
{Eoperatordef}).

Assumption~\ref{assbasic} ensures a solution to the generalized martingale
problem, as the following proposition, taken from
\cite{MR1326606}, Theorem 1.13.1, shows.
%
%
\begin{prop} 
Under Assumption~\ref{assbasic}, there is a unique solution
$(\qprob_x)_{x\in\hE}$ to the generalized martingale
problem on $E$. The family
$(\qprob_x)_{x\in\hE}$ possesses the strong Markov
property.
\end{prop}

Set $(\F_t)_{t \in\Real_+}$ to be the right-continuous
enlargement of $(\B_t)_{t\in\Real_{+}}$. Furthermore, with $\F$
denoting the smallest $\sigma$-algebra that contains $\bigcup_{t \in
\Real_+} \F_t$, we have $\F= \B_{\infty}$. Assumption
\ref{assbasic}
implies that
\[
f(X_{t\wedge\zeta_n}) - \int_0^{t\wedge\zeta_n}(Lf)(X_s)\,ds
\]
is a $(\Omega, (\F_t)_{t \in\Real_+} ,\qprob_{x}
)$-martingale for all
$n = 1,2,3,\ldots$ and $f\in C^{2}(E)$ since~$f$ and $Lf$ are bounded
on each $E_n$. By setting $f(x) = x^{i}, i=1,\ldots,d$, and $f(x) =
x^ix^j, i,j=1,\ldots,d$, it follows that, for each $n$ and each $x\in
\hE$,\break $X_{t\wedge\zeta_n}$~is~a~$(\Omega, (\F_t)_{t \in \Real_+}
,\qprob_{x})$-martingale with quadratic covariation
process\break $ \int_0^{\cdot} \indic_{\{t \leq\zeta_n\}}
c(X_t) \,\ud t$.

\subsection{Asymptotic growth rate}

For a fixed $x_0\in E$, set $\qprob= \qprob_{x_0}$. In the sequel,
whenever there is no subscript associated to the probabilities, it will
be tacitly assumed that they only charge the event $\{X_0 =
x_0\}$.\vadjust{\goodbreak}

Denote by $\Pi$
the class of probabilities on $(\Omega,\F)$ which are \textit{locally
absolutely continuous} with respect to $\qprob$ (written $\prob
\llloc\qprob$) and for which the coordinate process~$X$ does not
explode, that is,
%
%
\begin{equation}\label{EPidef}\qquad
\Pi= \bigl[P\in M_1(\Omega,\F) \dvtx\prob|_{\F_t} \ll
\qprob|_{\F_t} \mbox{ for all } t\geq0 \mbox{ and }\prob
[\zeta< \infty] = 0\bigr].
\end{equation}
For each $\prob\in\Pi$,~$X$ is a $(\Omega, (\F_t)_{t \in
\Real_+},
\prob)$-semimartingale such that $\prob[X\in C
(\Real_{+}, E) ] = 1$. Therefore,~$X$ admits the representation
\[
X = x_0 + \int_0^\cdot
b^\prob_t \,\ud t + \int_0^\cdot\sigma(X_t) \,\ud W^\prob_t,
\]
where $W^\prob$ is a standard $d$-dimensional Brownian motion on
$
(\Omega, (\F_t)_{t \in\Real_+}, \prob)$, $\sigma$~is the unique
symmetric strictly positive definite square root of $c$ and $b^\prob$
is a
$d$-dimensional $(\F_t)_{t \in\Real_+}$-progressively measurable process.

Let $(\xi_t)_{t\in\Real_+}$ be an adapted process. For $\prob\in
\Pi$, define
\[
\pliminft\xi_t :=\mathop{\operatorname{ess}\operatorname
{sup}}_{\prob}\Bigl\{ \chi\mbox{ is }
\F\mbox{-measurable}\bigm|\limt\prob[\xi_t \geq\chi] = 1\Bigr\}.
\]
If, in addition, $\prob[\xi_t > 0] = 1$ for each $t \in
\Real
_+$, let
\[
g(\xi; \prob) :=\sup\Bigl\{\gamma\in\Real\bigm|\pliminft
(t^{-1}\log\xi_t) \geq\gamma, \prob\mbox{-a.s.}
\Bigr\}
\]
be the \textit{asymptotic growth rate of $\xi$ under $\prob$}. Since
$\prob\in\Pi$ and
$\qprob$ are not necessarily equivalent on $\F$, $g(\xi; \prob)$
indeed depends
on $\prob\in\Pi$. The following result, the proof of which is
straightforward and hence omitted, provides an alternative
representation for $g(\xi;\prob)$.
%
%
\begin{lem}\label{Ltwogrowthsareone}
For a given $\prob\in\Pi$ and an adapted real-valued process $(\xi
_t)_{t\in\Real_{+}}$ such that $
\prob[\xi_t > 0] = 1$ for all $t \in\Real_+$,
\[
g(\xi;\prob) = \sup\Bigl\{\gamma\in\Real\bigm|\limt\prob
[t^{-1}\log\xi_t\geq\gamma]=1\Bigr\}.
\]
\end{lem}

%
%

\subsection{The problem}
The basic object in our study will be the class of all possible
nonnegative wealth processes that one can achieve by investing in the
$d$ assets whose price processes are modeled via~$X$. Whenever
$\vartheta$ is a $d$-dimensional predictable process, that is,
$X$-integrable under $\qprob$ (and, as a consequence,~$X$-integrable
under any $\prob\in\Pi$, as $\prob\llloc\qprob$), define the
process $V^{\vartheta} = 1 + \int_0^\cdot\vartheta_t' \,\ud X_t$, where
the prime symbol ($'$) denotes transposition throughout the text. Then
let $\V$ denote the class of all processes $V^\vartheta$ of the
previous form, where we additionally have $V^\vartheta\geq0$ up to
$\qprob$-evanescent sets. (Of course, $V^\vartheta\geq0$ also holds
up to $\prob$-evanescent sets for all $\prob\in\Pi$.) Naturally,
$\vartheta$ represents the position that an investor takes on the
assets whose discounted price-processes are given by~$X$, and
$V^\vartheta$ represents the resulting wealth from trading starting
from unit capital, constrained not to go negative at any time.\vadjust{\goodbreak}

The problem considered is to calculate
%
%
\begin{equation}\label{Erobustgrowthval} \sup_{V \in\V}
\inf_{\prob\in\Pi} g(V; \prob)
\end{equation}
and to find $\bV\in\V$ that attains this value, at least for all
$\prob$ in
a large sub-class of $\Pi$ that will be soon defined. To this
end, for a given $\lambda\in\Real$ and $L$ as in (\ref{Eoperatordef}),
define the cone of positive
harmonic functions with respect to $L+\lambda$ as
%
%
\begin{equation}\label{ECLdef}
H_{\lambda} :=\{\eta\in C^2(E) \such L\eta= -\lambda\eta
\mbox
{ and } \eta
> 0\}.
\end{equation}
Set
%
%
\begin{equation}\label{Eblambdadef}
\blambda:=\sup\{\lambda\in\reals\such H_\lambda\neq
\varnothing
\}.
\end{equation}
Since $H_0\neq\varnothing$ (take $\eta
\equiv1$), it follows that $\blambda\geq0$. If $H_{\blambda}\neq
\varnothing$, then, by
construction, there is an $\breta\in C^2(E)$ satisfying
%
%
\begin{equation}\label{eqPDE}
L\breta= -\blambda\breta,
\end{equation}
and $\blambda$ is the largest real for which such an $\breta$ exists.
Thus $\blambda$ is a generalized version of the principal eigenvalue
for $L$ on
$E$. The following result, taken from~\cite{MR1326606}, Theorem 4.3.2,
states that, indeed, $H_{\blambda}\neq\varnothing$.
%
%
\begin{prop}\label{PLambdastructure}
Let Assumption~\ref{assbasic} hold. Then $0\leq\blambda< \infty$ and
$H_{\blambda}\neq\varnothing$.
\end{prop}
%
%
\begin{rem}
To connect Proposition~\ref{PLambdastructure} with
\cite{MR1326606}, Theorem 4.3.2, note that $\lambda_c(D)$ therein is equal to
$-\blambda$. Note also that, by its construction, $\Pi= \varnothing$
if there exists a $t >
0$ such that $\qprob[\zeta> t] = 0$. However, by
\cite{MR1326606}, Theorem~4.4.4, it follows that if such a $t>0$ exists, then
$\blambda=
\infty$. Proposition
\ref{PLambdastructure} thus implies that $\qprob[\zeta>
t] > 0$ for all $t>0$. It is also directly shown in the proof of
Theorem~\ref{thmasymptogrowth1}
below that under Assumption~\ref{assbasic}, $\Pi\neq\varnothing$.
\end{rem}
%
%
\begin{rem}
Proposition~\ref{PLambdastructure} makes no claim regarding the
uniqueness of
$\breta$ corresponding to $\blambda$. For example, when $E=(0,\infty
)$ and
$c\equiv1$, it holds that \mbox{$\blambda= 0$}; hence $\breta$ could be
either~$X$
or $1$. For this $E$ and $c$, Example~\ref{Exaqexpllam0} in Section
\ref{Srelarb} shows that even when uniqueness fails,
a particular choice of $\breta$ may be advantageous.
\end{rem}

The following result, taken from~\cite{MR1326606},
Theorems 4.3.3 and 4.3.4,
provides a way of checking if a particular pair $(\eta,\lambda)$ such that
$\eta\in H_{\lambda}$ corresponds to an optimal pair $(\breta
,\blambda
)$ and
if the optimal pair is unique.
%
%
\begin{prop}\label{Pblambdatest}
Let Assumption~\ref{assbasic} hold. Let $(\eta, \lambda)$ be such that
$\eta\in H_{\lambda}$. Then there exists a unique solution
$(\prob^{\eta}_x)_{x\in\hE}$ to the generalized martingale
problem on $\hE$ for the operator
%
%
\begin{equation}\label{Edoobetatransform}
L^{\eta} = L + c\nabla\log\eta\cdot\nabla,\vadjust{\goodbreak}
\end{equation}
and $(\prob^{\eta}_x)_{x\in\hE}$ possesses the strong Markov
property. Furthermore, if the coordinate mapping process~$X$ is
recurrent under
$(\prob^{\eta}_x)_{x\in E}$, then $\eta$ is unique up to
multiplication by a positive constant, $\breta= \eta$ and $\blambda=
\lambda$.
\end{prop}
%
%
\begin{rem}
Proposition~\ref{Pblambdatest} only covers the case where the
coordinate mapping process~$X$ is recurrent under $(\prob^{\eta
}_x)_{x\in E}$. It should be noted, however, that even when the
coordinate mapping process~$X$ under $(\prob^{\eta}_x
)_{x\in
E}$ is transient, $\eta= \breta$ and $\lambda= \blambda$ is still
possible. Indeed, in Example~\ref{Exaqexpllam0} from Section~\ref
{Srelarb}, $\blambda= 0$
even though $\qprob_x[\zeta< \infty] > 0$ for all $x\in E$, and thus
$\breta= 1$ does not yield a recurrent process.
\end{rem}

\section{The min--max result}\label{Sminmaxresult}

\subsection{The result}

For future reference, let $\breta$ be a solution of (\ref{eqPDE})
corresponding to $\blambda$ with $\breta(x_0) = 1$, and define the
function $\bell\dvtx E \mapsto\Real$ via
%
%
\begin{equation}\label{Ebelldef}
\bell(x) = \log\breta(x)\qquad  \mbox{for } x \in E.
\end{equation}

The following result identifies $\blambda$ with the value in (\ref
{Erobustgrowthval}).
%
%
\begin{theorem} \label{thmasymptogrowth1}
Let Assumption~\ref{assbasic} hold. Let $\breta$ be a solution of
(\ref{eqPDE}) corresponding to
$\blambda$ with $\breta(x_0) = 1$, and define $\bV$ via $\bV_t =
e^{\blambda t}   \breta(X_t)$ for all $t \in\Real_+$. Define also
\[
\bPi:=\Bigl\{\prob\in\Pi\bigm|\pliminft(t^{-1}\log
\breta(X_t)) \geq0 ,\prob\mbox{-a.s.}\Bigr\}.
\]
Then $\bV\in\V$ and $g(\bV; \prob) \geq\blambda$ for all $\prob
\in
\bPi$. Furthermore,
%
%
\begin{equation} \label{eqminmaxforasymptogrowth}
\blambda= \sup_{V \in\V} \inf_{\prob\in\bPi} g(V; \prob) =
\inf_{\prob\in\bPi} \sup_{V \in\V} g(V; \prob).
\end{equation}
\end{theorem}
%
%
\begin{rem}\label{Rtightfamilyinclusion}
The normalized eigenfunction $\breta$ in the statement of Theorem~\ref
{thmasymptogrowth1} may not be unique. Since the class of measures
$\bPi $ depends upon~$\breta$, the variational problems in (\ref
{eqminmaxforasymptogrowth}) also change with~$\breta$. However, the
value~$\blambda$ is the same no matter which~$\breta$ is chosen.

For a given $\breta$, it may seem artificial to restrict attention to
$\bPi$. However, no matter which
$\breta\in H_{\blambda}$ is chosen, $\bPi$ contains all the
probabilities $\prob$ such
that~$X$ is tight in $E$, and hence naturally corresponds
to those $\prob$ for which~$X$ is stable. To see this, assume that~$X$
is tight, and let
$\epsilon> 0$ and $K^{\epsilon}\subseteq E$ be compact such that
$\sup
_{t\geq
0}\prob[X_t\notin K^{\epsilon}]\leq\epsilon$. Set
$\beta
^{\epsilon} =
{\max_{x\in K^{\epsilon}}} |{\log\breta}(x)|$, and note that for any
$\delta> 0$
and $t > \beta^{\epsilon} / \delta$,
\[
\prob[t^{-1}\log\breta(X_t)< -\delta] \leq
\prob[|t^{-1}\log\breta(X_t)| > \delta;X_t\notin
K^{\epsilon} ]\leq\epsilon.
\]
Thus, $\limt\prob[t^{-1}\log\breta(X_t)\geq-\delta]
= 1$ for all
$\delta
> 0$; hence, $\prob\in\Pi^{*}$.
\end{rem}
\begin{pf*}{Proof of Theorem~\ref{thmasymptogrowth1}}
To see why $\bV\in\V$, note that It\^o's formula gives, for each
$\nin
$, each $t \in\Real_+$ and each $\prob\in\Pi$,
%
%
\begin{eqnarray}\label{EVrep}
\bV_{t\wedge\zeta_n} &=& 1 + \int_0^{t\wedge\zeta_n} e^{\blambda s}
\nabla\breta(X_s)'\,\ud X_s\nonumber\\[-8pt]\\[-8pt]
&=& 1 + \int_0^{t\wedge\zeta_n}\bV_s
\nabla\bell(X_s)'\,\ud X_s.\nonumber
\end{eqnarray}
Since $
\prob[\zeta< \infty]=0$ for all $\prob\in\Pi$, it follows that the
equalities in (\ref{EVrep}) hold under $\prob$ when we replace $t
\wedge\zeta_n$ with $t$ for all $t \in\Real_+$. By the
construction of $\bPi$,
$\prob[\limt t^{-1}\log(\bV_t)\geq\gamma] =1$ holds
for all
$\gamma< \lambda^*$ and all
$\prob\in\bPi$. Therefore, Lemma~\ref{Ltwogrowthsareone} implies
$g(\bV; \prob) \geq\blambda$ for all
$\prob\in\bPi$. In particular, $\blambda\leq\sup_{V \in\V}
\inf_{\prob\in\bPi} g(V; \prob)$.

Now, let $\blambda_n$, $\breta_n$ and $\bell_n$ be the equivalents of
$\blambda$, $\breta$ and $\bell$ when $E$ is replaced by $E_n$ in
(\ref{ECLdef}), (\ref{Eblambdadef}), (\ref{eqPDE}) and (\ref
{Ebelldef}). Assumption
\ref{assbasic} gives that $c$ is uniformly elliptic on $E_n$ and hence
$\breta_n \in C^{2,\alpha}(\bar{E}_n)$ and vanishes on
$\partial E_n$~\cite{MR1326606}, Theorem 3.5.5. Furthermore, there exists a solution
to the
generalized martingale problem $(\prob^{*}_{x,n})_{x\in
E_n}$ for
the operator $L^{\breta_n}$ in (\ref{Edoobetatransform}) and the
coordinate process~$X$ under
$(\prob^{*}_{x,n})_{x\in E_n}$ is recurrent in $E_n$
(\cite{MR1326606}, proof of Theorem~4.2.4). This latter fact gives the
uniqueness (up to multiplication by a positive constant) of $\breta_n$.

Set $\bprob_n\!=\!\bprob_{x_0,n}$. It follows that $\bprob_n
[\zeta\!<\!\infty]\!=\!0$ and $\limt\bprob_n [t^{-1} \log\breta(X_t)\!=\break 0 ]\!=\!1$ since
there exists a $K_n>0$ such that $1/K_n < \breta< K_n$ on $E_n$. Thus,
$\bprob_n \in\bPi$ if $\bprob_n \llloc\qprob$. To\vspace*{-1pt} show
the latter, let $(\qprob_{x,n})_{x\in \hE_n}$ be the solution to the
generalized martingale problem for $L$ on $\hE_n$. Let $\qprob_n =
\qprob_{x_0,n}$. It follows from~\cite{MR1326606}, Corollary 4.1.2, and
the recurrence of~$X$ under $\prob^{*}_n$ that for $t>0$,
%
%
\begin{equation}\label{Epnqnrnderiv}
\frac{d\prob^{*}_n}{d\qprob_n}\bigg|_{\mathcal{B}_t} =
e^{\blambda_nt}\frac{\breta_n(X_t)}{\breta_n(x_0)} \indic_{\{
\zeta_n > t\}},
\end{equation}
and thus $\bprob_n |_{\B_t} \ll\qprob_n |_{\B_t}$. This
immediately gives
$\bprob_n |_{\B_{t\wedge\zeta_n}} \ll\qprob_n |_{\B_{t\wedge
\zeta_n}}$ for
each $n$. But, $\qprob_n |_{\B_{t\wedge\zeta_n}} = \qprob
|_{\B_{t\wedge\zeta_n}}$. If $B\in\mathcal{B}_t$ is such that
$\qprob
[B] =
0$, then $\qprob[B\cap\{\zeta_n > t\}] = 0$. Since
$B\cap\{\zeta_n > t\}\in
\mathcal{B}_{t\wedge\zeta_n}$, it follows that $\bprob_n [B\cap
\{\zeta_n > t\}] =
0$. But, $\bprob_n [\zeta_n > t ] = 1$ for each $t$ so $\bprob_n
[B\cap
\{\zeta_n > t\}] =
0$ implies $\bprob_n [B] =0$. Therefore, $\bprob_n|_{\mathcal{B}_t}
\ll
\qprob|_{\mathcal{B}_t}$ and hence $\bprob_n|_{\mathcal{F}_t} \ll
\qprob|_{\mathcal{F}_t}$ as well, proving
$\bprob_n\in\bPi$.

Let $V_n^*$ be defined via $V_n^*(t) = e^{\blambda_n(t\wedge\zeta_n)}
\eta^*_n(X_{t\wedge\zeta_n})$ for $t \in\Real_+$ [in order to avoid the
cumbersome notation $V_{n, t}^*$ for $t \in\Real_+$, we simply use
$V_{n}^* (t)$ here]. The same computations as in (\ref{EVrep})
show that, for all $\prob\in\Pi$,
\[
V_n^* = 1+ \int_0^\cdot\indic_{\{t \leq\zeta_n\}
}e^{\lambda_n
t}\nabla\breta_n(X_t)'\,\ud X_t
\]
and hence $V_n^* \in\V$. Note that $V_n^*$ stays strictly positive under
$\prob_n^*$ since $\bprob_n[\zeta_n < \infty] = 0$.
Now, $g(V_n^*;
\bprob_n) \leq\blambda_n$ is
immediate since $E_n$ is bounded, and hence $\eta^*_n$ is bounded above
on $E_n$. Furthermore,
$V_n^*$ is the \textit{num\'eraire portfolio in $\V$ under $\bprob_n$}, which
means that $V/ V_n^*$ is a (nonnegative) $\prob_n^*$-supermartingale
for all
$V \in\V$. To wit, consider any other $V \in\V$, and write $V = 1 +
\int_0^\cdot\vartheta'_t \,\ud X_t$. A straightforward use of It\^o's formula
using the fact that $L \breta_n(x) = -\blambda_n(x) \breta_n(x)$ holds
for all
$x \in E_n$ gives that, under $\prob_n^*$,
\[
\frac{V }{V_n^* } = \int_0^\cdot\biggl(\frac{\vartheta_t - V_t
\nabla\bell_n(X_t)}{V_n^*(t)}\biggr)'\,d \bigl(X_t -
c(X_t)\nabla\bell
_n(X_t)\,\ud t\bigr);
\]
since the process $X - \int_0^\cdot c(X_t)\nabla\bell_n(X_t)\,\ud t$ is a
local $\bprob_n$-martingale, the num\'e\-raire property of $V_n^*$ in $\V$
under $\bprob_n$ follows. In view of the nonnegative supermartingale
convergence theorem, the nonnegative supermartingale property of $V/
V^*_n$ under $\bprob_n$ gives that $\limsup_{t \to\infty} \log
(V_t / V^*_n(t)) \leq0$ in the $\bprob_n$-a.s. sense.
Therefore, $g(V;
\bprob_n) \leq g(V_n^*; \bprob_n)$ holds for all $V \in\V$.
Since $g(V_n^*; \bprob_n) \leq\blambda_n$, $\sup_{V \in\V} g(V;
\bprob
_n) \leq\blambda_n$ holds, and
$\inf_{\prob\in\Pi^*} \sup_{V \in\V} g(V; \prob) \leq\inf
_{n\in
\Natural}
\blambda_n$.\break However, $\mbox{$\downarrow$}\limn\blambda_n = \blambda$ holds
in view of Assumption~\ref{assbasic} (\cite{MR1326606}, Theorem~4.4.1).
This gives $\inf_{\prob\in
\Pi^*} \sup_{V \in\V} g(V; \prob) \leq\blambda$ and completes the
argument.
\end{pf*}

\subsection{An ``almost sure'' class of measures}

For a fixed $\breta\in H_{\blambda}$, define the following class of
probability measures:
\[
\bPi_{\mathrm{a.s.}} :=\Bigl\{\prob\in\Pi\bigm|\liminft
(t^{-1}\log\breta(X_t)) \geq0, \prob\mbox{-a.s.}
\Bigr\}.
\]
It is straightforward to check that $\bPi_{\mathrm{a.s.}}\subseteq\bPi$. Furthermore,
as will be seen in Section~\ref{Sbprob}, it can be easier to verify
inclusion in $\bPi_{\mathrm{a.s.}}$ than
$\bPi$. For $\prob\in\Pi$ and $V\in\V$ define
\[
g_{\mathrm{a.s.}} (V; \prob) :=\sup\Bigl\{\gamma\in\Real\bigm|\liminft
(t^{-1}\log V_t)\geq\gamma, \prob
\mbox{-a.s.}\Bigr\}
\]
as the ``almost sure'' growth of the wealth $V$. The following result
is the analog
of Theorem~\ref{thmasymptogrowth1} for the class of measures $\bPi
_{\mathrm{a.s.}}$ and for the growth rate
$g_{\mathrm{a.s.}} (V;\prob)$.
%
%
\begin{prop} \label{Palmostsureresult}
Let Assumption~\ref{assbasic} hold. Let $\breta\in H_{\blambda}$ be such
that $\breta(x_0) = 1$, and define $\bPi_{\mathrm{a.s.}}$ as above. Define
$\bV\in\V$ by $\bV_t = e^{\blambda t}\breta(X_t)$,\break $t\geq0$, as in Theorem
\ref{thmasymptogrowth1}. Then
$g_{\mathrm{a.s.}} (\bV;\prob)\geq\blambda$ for all $\prob\in\bPi_{\mathrm{a.s.}}$ and
\[
\blambda= \sup_{V\in\V}\inf_{\prob\in\bPi_{\mathrm{a.s.}}} g_{\mathrm{a.s.}} (V;
\prob)
= \inf_{\prob\in\bPi_{\mathrm{a.s.}}}\sup_{V \in\V} g_{\mathrm{a.s.}} (V; \prob).
\]
\end{prop}
%
%
\begin{rem}
Concerning the class $\bPi$, in Remark~\ref{Rtightfamilyinclusion} it was
discussed that when the coordinate process~$X$ is $\prob$-tight, then
$\prob\in\bPi$. In contrast, a~useful characterization of even a
subset of
$\bPi_{\mathrm{a.s.}}$ independent of $\breta$ is difficult. On the positive
side, if
$\prob$ is such that~$X$ never exits $E_n$ for some $n$, then $\prob
\in
\bPi_{\mathrm{a.s.}}$.
However, even if~$X$ is positive recurrent under $\prob$, it cannot
immediately be said that $\prob\in\bPi_{\mathrm{a.s.}}$
\end{rem}
\begin{pf*}{Proof of Proposition~\ref{Palmostsureresult}}
$\!\!\!$By construction of the class $\bPi_{\mathrm{a.s.}}$ it follows that $g_{\mathrm{a.s.}}
(\bV;
\prob) \geq\blambda$ for all $\prob\in\bPi_{\mathrm{a.s.}}$. Thus
$\blambda\leq
\sup_{V \in\V} \inf_{\prob\in\bPi_{\mathrm{a.s.}}} g_{\mathrm{a.s.}} (V; \prob)$.
The inequality
$\blambda\geq\inf_{\prob\in\bPi_{\mathrm{a.s.}}} \sup_{V \in\V}g_{\mathrm{a.s.}} (V;
\prob)$
follows by essentially the same argument as in Theorem
\ref{thmasymptogrowth1}. Specifically, let $\blambda_n, \breta_n,
\bell
_n, \bV_n$
and $\bprob_n$ be as in the proof of Theorem~\ref{thmasymptogrowth1}. It
was shown therein that $\bprob_n\in\Pi$ for each $n$ and that the coordinate
process is recurrent in $E_n$ under $(\bprob_{x,n})_{x\in E_n}$. In fact,
$\bprob_n\in\bPi_{\mathrm{a.s.}}$ because there is a $K_n > 0$ such that
$1/K_n <
\breta
< K_n$ on $E_n$ and hence, $\bprob_n$-a.s., $\lim_{t\rightarrow
\infty}
t^{-1}\log\breta(X_t) = 0$. Furthermore, since $\breta_n$ is bounded
from above on
$E_n$ it holds that $g_{\mathrm{a.s.}}(\bV_n;\bprob_n)\leq\blambda_n$. Using the
num\'eraire property of $\bV_n$ under $\bprob_n$ and the supermartingale
convergence
theorem, it follows that $g_{\mathrm{a.s.}}(V;\bprob_n)\leq
g_{\mathrm{a.s.}}(\bV;\bprob_n)$ holds for all $V\in\V$. Therefore,
$\sup_{V\in\V}g_{\mathrm{a.s.}}(V;\bprob_n)\leq\blambda_n$ and
\[
\inf_{\prob\in\bPi_{\mathrm{a.s.}}}\sup_{V\in\V}g_{\mathrm{a.s.}}(V;\bprob_n)\leq
\inf_{n\in\Natural}\blambda_n = \blambda
\]
since $\mbox{$\downarrow$}\lim_{n\rightarrow\infty}\blambda_n = \blambda$ as
seen in
the proof of Theorem~\ref{thmasymptogrowth1}. This completes the argument.
\end{pf*}

\section{An interesting probability measure}\label{Sbprob}

Let $\breta\in H_{\blambda}$, and let $(\bprob_x)_{x\in
\hE
}$ be the solution to the
generalized martingale problem on $\hE$ for the operator~$L^{\breta}$
given in~(\ref{Edoobetatransform}). Set $\bprob\equiv\bprob_{x_0}$.

It is of great interest to know whether $\prob^* \in\Pi^*$. To
begin with, if this is indeed true and $g(V^{*},\bprob)=\blambda$,
the pair $(V^*, \prob^*)$ constitutes a saddle point for the minimax
problem described in (\ref{eqminmaxforasymptogrowth}). Indeed, using
the num\'eraire property of $\bV$ under
$\bprob$ and the definition of $\bPi$, it follows that, in this case
(see the
proof of Theorem~\ref{thmasymptogrowth1})
\[
g (V; \prob^*) \leq g (V^*; \prob^*) \leq g (V^*; \prob)\qquad \mbox{for
all } V \in\V\mbox{ and } \prob\in\Pi^*.
\]
Furthermore, in Section~\ref{Srelarb} where connections between robust
growth-optimal portfolios and optimal arbitrages are studied, the
behavior of the coordinate process~$X$ under $\bprob$ becomes
important. To this
end, presented in the sequel are some results that explore the behavior
of~$X$
under $\bprob$. In particular, Propositions~\ref{Phopfprop} and \ref
{PQtailsprop1} give sufficient conditions to ensure that $\prob^* \in
\Pi^*$.
%
%
\begin{rem}
By construction, if $\bprob[\zeta< \infty] > 0$ then
$\bprob
\notin\bPi$. Example~\ref{Exaqexpllam0} provides a case when explosion of
$X$ under $\bprob$ occurs \textit{for some} $\breta\in H_{\blambda}$.
Furthermore,~\cite{Pinchover95} contains an example showing that for
\textit{for all} $\breta\in H_{\blambda}$, the probability $\bprob$ that
is constructed from $\breta$ leads to explosive behavior of~$X$ under
$\bprob$; hence, none of the candidate $\bprob$ is in $\bPi$. Now,
consider the case when $\breta\in H_{\blambda}$ is such that~$X$ is
nonexplosive under $\bprob$. In this instance, Corollary \ref
{CQtailsprop1} shows that if $\blambda= 0$, then $\bprob\in\bPi$. As
for when $\blambda> 0$, although only \textit{sufficient} conditions
ensuring that $\prob^* \in\Pi^*$ are presented in this section,
examples where $\prob^* \notin\Pi^*$ have not been found. It is an
open question whether, under Assumption~\ref{assbasic}, $\bprob\in
\bPi$
holds whenever $\blambda> 0$ and $\breta\in H_{\blambda}$ is such that
$X$ is nonexplosive under $\bprob$. See Example~\ref{Exaimpossible} in
Section~\ref{Sexamples} for a potential counterexample.
\end{rem}

The first result gives conditions under which $\bprob\in\Pi$ and
relates the tail probabilities of $\zeta$ under $\qprob$ and robust
growth-optimal strategies.
%
%
\begin{prop}\label{Prelarbrobgrowthprop}
Let Assumption~\ref{assbasic} hold, and let $\breta\in H_{\blambda}$ be
such that $\bprob_x[\zeta< \infty] = 0$ holds for all $x\in E$. Then
$\bprob\in\Pi$ and
%
%
\begin{equation}\label{Eqprobexitrep}\qquad
\qprob_{x} [\zeta> T] =
\breta(x)\mathbb{E}^{\bprob}_{x}\biggl[\frac{1}{V^*_T}\biggr]
\qquad\mbox{holds for all }
T \in\Real_+ \mbox{ and } x \in E.
\end{equation}
\end{prop}
\begin{pf}
In a similar manner to (\ref{Epnqnrnderiv}), if $\bprob[\zeta <\infty]
=0 $, then it follows from~\cite{MR1326606}, Corollary 4.1.2, that
\[
\frac{d\bprob}{d\qprob}\bigg|_{\mathcal{B}_t} = e^{\blambda
t}\frac{\breta(X_t)}{\breta(x_0)}\indic_{\{\zeta> t\}}
\]
from which it immediately holds that $\bprob\llloc\qprob$, and hence
$\bprob\in\Pi$.
Given $V^*_T = \exp(\lambda^* T) \eta^* (X_T)$, the equality in
(\ref{Eqprobexitrep}) follows immediately from
\cite{MR1326606}, Theorem~4.1.1.
\end{pf}

Recall from Remark~\ref{Rtightfamilyinclusion} that $\prob^*$-tightness
of $(X_t)_{t\in\Real_+}$ implies that \mbox{$\prob^* \in
\Pi^*$}. The following result is useful because it shows that, under
Assumption~\ref{assbasic}, positive recurrence and tightness of
$(X_t)_{t \in\Real_+}$ under $\prob^*$ are equivalent notions. Note
that, in general,
even in the one-dimensional bounded case, the behavior of $(X_t)_{t \in
\Real_+}$ under $\prob^*$ can vary from positive recurrence to
transience as
is shown in the examples in Section~\ref{SSonedimex}.
%
%
\begin{prop}\label{Ptighteqposrec}
Let Assumption~\ref{assbasic} hold. Then the following are
equivalent:
\begin{longlist}[(2)]
\item[(1)] The coordinate mapping process~$X$ is positive recurrent
under $(\bprob_x)_{x\in
E}$.
\item[(2)] For some $x\in E$ the family of random variables
$(X_t)_{t\geq0}$ is $\bprob_x$-tight in~$E$.
\end{longlist}
\end{prop}
\begin{pf}
Under Assumption~\ref{assbasic},~$X$ is recurrent under $(\bprob
_x)_{x\in E}$ if for any $x, y \in E$ and $\eps> 0$, if
$\tau_{B(y,\eps)}$ is the first time the coordinate process enters into
the closed ball of radius $\eps$ around $y$, then
$\bprob_x[\tau_{B(y,\eps)} < \infty] = 1$. Note that if~$X$ is
recurrent under $(\bprob_x)_{x\in E}$, then for all $x\in E$,
$\bprob_x[\zeta< \infty] = 0$ (\cite{MR1326606}, Theorem~2.8.1).
Furthermore, given that~$X$ is recurrent under $ (\bprob_x)_{x\in E}$,
then~$X$ is further positive recurrent under $(\bprob_x)_{x\in E}$ if
there exists a function $\tilde{\eta}^{*} > 0$ such that $\tilde
{L}^{*}\tilde {\eta}^{*} = 0$ and $\tilde{\eta}^{*}\in\Lb^1(E,
\mathrm{Leb})$ where $\tilde {L}^{*}$ is the formal adjoint to $L^{*}$
(\cite{MR1326606}, Section 4.9). Under Assumption~\ref{assbasic}, and
recalling the definition of $\bell$ from (\ref{Ebelldef}),
$\tilde{L}^{*}$ is the differential operator acting on $f\in C^2(E)$ by
\[
\tilde{L}^{*} f (x) =
\frac{1}{2}\sum_{i,j=1}^d\frac{\partial^2}{\partial
x_i\,\partial x_j}(c_{ij}(x)f(x)) -
\sum_{i=1}^d\frac{\partial}{\partial
x_i}((c(x)\nabla\bell(x))_{i}f(x)).
\]

Assume that~$X$ is positive recurrent under
$(\bprob_x)_{x\in E}$ and normalize $\tilde{\eta}^{*}$ so
that $\int_{E}\tilde{\eta}^{*}(y)\,dy = 1$. By the ergodic theorem
(\cite{MR1326606}, Theorem 4.9.9) it follows that for any bounded measurable function
$f\dvtx E\mapsto\reals$,
%
%
\begin{equation}\label{Eposrecergodicthm}
\lim_{t\uparrow\infty}\mathbb{E}^{\bprob}_{x}[f(X_t)
] =
\int_{E}f(y)\tilde{\eta}^{*}(y)\,dy.
\end{equation}
Since $\tilde{\eta}^{*}$ is a probability density, for any $\eps>0$
there is a compact set
$K_{\eps}\subset E$ such that
\[
\int_{K_{\eps}^c}\tilde{\eta}^{*}(y)\,dy \leq\eps.
\]
Thus, taking $f_{\eps}(x) = \indic_{K_\eps^c}(x)$ in
(\ref{Eposrecergodicthm}), the continuity of~$X$ and $\bprob
[\zeta< \infty]=0$ imply that $(X_t)_{t\geq0}$
is $\bprob_x$-tight
for any $x\in E$.

As for the reverse implication, assume for some $x\in E$ that
$(X_t)_{t\geq0}$ is $\bprob_x$-tight in $E$,
and for each $\eps$ let $K_{\eps}\subset E$ be a compact set such that
%
%
\begin{equation}\label{Ebprobxtightcompact}
\inf_{t\geq0}\bprob_x [X_t\in K_\eps] \geq
1-\eps.
\end{equation}

Under Assumption~\ref{assbasic} there are only three possibilities
for the coordinate process~$X$ under $(\bprob_x)_{x\in
\hat{E}}$ (\cite{MR1326606}, Section 2.2.8):
\begin{longlist}[(3)]
\item[(1)]~$X$ is transient: for all $x\in E$ and $n\in\Natural$,
$\bprob_x[X \mbox{ is eventually in } E_n^c] = 1$;
\item[(2)]~$X$ is null recurrent:~$X$ is recurrent and for any $\phi
\in C^2(E),
\phi> 0$ such that
$\tilde{L}^{*}\phi= 0$, $\int_{E}\phi(y)\,dy = \infty$;
\item[(3)]~$X$ is positive recurrent:~$X$ is recurrent but not null recurrent.
\end{longlist}
Clearly, if $(X_t)_{t\geq t_0}$ is $\bprob_x$-tight in $E$
for some $x\in E$, then~$X$ cannot be transient.
Furthermore, if~$X$ were null recurrent, then for each $x\in E$ and
any compact set $K\subset E$ it would follow that
(\cite{MR1326606}, Theorem 4.9.5)
\[
\lim_{t\uparrow\infty}\frac{1}{t}\int_0^t\bprob_x[X_s\in
K]\,ds = 0.
\]
But, by the assumption of tightness, for the compact set $K_{\eps
}\subset E$
appearing in (\ref{Ebprobxtightcompact}),
\[
\liminf_{t\uparrow\infty}\frac{1}{t}\int_0^t\bprob_x
[X_s\in K_\eps]\,ds \geq(1-\eps).
\]
Therefore,~$X$ cannot be null-recurrent. Thus~$X$ is positive
recurrent under $(\bprob_x)_{x\in E}$.\vadjust{\goodbreak}
\end{pf}

The following result is useful when point-wise estimates for
$\breta$ are available.
%
%
\begin{prop}\label{Phopfprop}
Let Assumption~\ref{assbasic} hold, and let $\bell$ be as in (\ref
{Ebelldef}). If \mbox{$\blambda>0$}, $\bprob[\zeta<\infty] =
0$ and
%
%
\begin{equation}\label{Ehopfcond1}
\lim_{n\uparrow\infty}\inf_{x\in E_n^c}\frac{1}{2}\nabla\bell(x)'
c(x)\nabla\bell(x) \geq\blambda,
\end{equation}
then $\bprob\in\bPi$.
\end{prop}
%
%
\begin{rem}
If $c$ is uniformly elliptic on $E$, and $E$ is bounded with a~smooth
boundary, $\blambda$ corresponds to the principal eigenvalue for $L$
acting on
functions $\eta$ which vanish on $\partial E$. Since $
(e^{\blambda
t}\breta(X_t))^{-1}$ is a $\bprob$-supermartin\-gale, it follows
that $\bprob[\zeta< \infty] = 0$. Furthermore, Hopf's
lemma asserts that~$\nabla\breta$ does not vanish on~$\partial E$, so (\ref{Ehopfcond1})
holds as
well; indeed, the quantity on the left-hand side is unbounded from above.
\end{rem}
\begin{pf*}{Proof of Proposition~\ref{Phopfprop}}
That $\bprob\in\Pi$ follows by Proposition
\ref{Prelarbrobgrowthprop}.
Recall that $\breta(x_0) =1$. Now,
%
%
\begin{eqnarray}\label{Ebproblogrep}
\frac{1}{t}\bell(X_t) &=&
\frac{1}{t}\int_0^t\biggl(\frac{1}{2}\nabla\bell(X_s)'c(X_s)\nabla
\bell
(X_s)-\blambda\biggr) \,d s\nonumber\\[-8pt]\\[-8pt]
&&{}+ \frac{1}{t}\int_0^t\nabla\bell(X_s)'\sigma(X_s) \,\ud
W^{\bprob}_s,\nonumber
\end{eqnarray}
where $W^{\bprob}$ is a Brownian motion under $\bprob$. By
(\ref{Ehopfcond1}), there is a $\tilde{\lambda} > 0$ such that, for $n$
large enough,
%
%
\begin{eqnarray}\label{Ehopfgradintbound}\qquad
&&\int_0^t\nabla\bell(X_s)'c(X_s)\nabla\bell(X_s) \,d s \nonumber\\[-8pt]\\[-8pt]
&&\qquad\geq
\tilde{\lambda}\int_0^t\indic_{\{X_s\in E_n^c\}}\,ds +
\int_0^t\nabla\bell(X_s)'c(X_s)\nabla\bell(X_s)\indic_{\{X_s\in
E_n\}}
\,d s.\nonumber
\end{eqnarray}
Under Assumption~\ref{assbasic},~$X$ is either positive recurrent, null
recurrent or transient under $(\bprob_x)_{x\in E}$. If
$X$ is
positive recurrent, then, since $\blambda> 0$ implies that $\breta$
is not
identically constant, it follows that (\cite{MR1326606}, Theorem 4.9.5)
for $n$
large enough
\[
\lim_{t\uparrow\infty}\int_0^t\nabla\bell(X_s)'c(X_s)\nabla\bell
(X_s)\indic_{\{X_s\in
E_n\}} \,d s = \infty,\qquad \bprob\mbox{-a.s.}
\]
Similarly, if~$X$ is either null recurrent or transient it follows that (again,
by~\cite{MR1326606}, Theorem 4.9.5)
\[
\lim_{t\uparrow\infty}\tilde{\lambda}\int_0^t\indic_{\{X_s\in
E_n^c\}
}\,ds = \infty,\qquad \bprob\mbox{-a.s.}
\]
Using (\ref{Ehopfgradintbound}) it thus holds in each case
\[
\lim_{t\uparrow\infty}
\int_0^t\nabla\bell(X_s)'c(X_s)\nabla\bell(X_s) \,d s = \infty
,\qquad
\bprob\mbox{-a.s.}
\]
Let
$M = \int_0^\cdot\nabla\bell(X_s)'\sigma(X_s) \,\ud W^{\bprob}_s$,
so that
\[
[M, M] = \int_0^\cdot\nabla\bell(X_s)'c(X_s)\nabla\bell(X_s)
\,\ud s.
\]
By the Dambins, Dubins and Schwarz theorem
(\cite{MR1121940}, Theorem 3.4.6), there exists a standard Brownian motion (under
$\bprob$) $B$ such that $M = B_{[M, M]_\cdot}$. Therefore, one can
write (\ref{Ebproblogrep}) as
\[
\frac{1}{t}\bell(X_t) = - \lambda^* + \frac{[M, M]_t}{2 t} \biggl(1
+ 2 \frac{B_{[M, M]_t}}{[M, M]_t}\biggr).
\]
By the strong law of large numbers,
\[
\lim_{t\uparrow\infty}\frac{B_{[M, M]_t}}{[M, M]_t} = 0,\qquad \bprob
\mbox{-a.s.},
\]
which means that
%
%
\begin{equation}\label{Eliminfequiv}
\liminf_{t\uparrow\infty}\frac{1}{t}\bell(X_t) \geq
- \lambda^* + \liminf_{t\uparrow\infty} \frac{[M,M]_t}{2 t},\qquad
\bprob\mbox{-a.s.}
\end{equation}

If~$X$ is positive recurrent under $\bprob$, then $\bprob\in\bPi$ as
shown in
Proposition~\ref{Ptighteqposrec} and Remark
\ref{Rtightfamilyinclusion}. Otherwise, note that
because of (\ref{Ehopfcond1}), for any $\delta> 0$ and $\nin$ large
enough,
\begin{eqnarray*}
- \lambda^* + \frac{[M,M]_t}{2 t} &\geq&
- \delta\frac{1}{t}\int_0^t \indic_{\{X_s\in E_n^c\}} \,d s -
\blambda\frac{1}{t}\int_0^t \indic_{\{X_s\in E_n\}} \,d s\\
&\geq& - \delta-\blambda\frac{1}{t}\int_0^t \indic_{\{X_s\in E_n\}}
\,d s.
\end{eqnarray*}
Now, if~$X$ is null-recurrent under $\bprob$, then from
\cite{MR1326606}, Theorem 4.9.5, it follows that
\[
\lim_{t\uparrow\infty}\frac{1}{t}\int_0^t \indic_{\{X_s\in E_n\}}
\,d s
= 0, \qquad \bprob\mbox{-a.s.}
\]
proving, in view of (\ref{Eliminfequiv}), that $\bprob\in\bPi_{\mathrm{a.s.}}$,
and hence $\bprob\in\bPi$. Clearly,
\[
\{X \mbox{ eventually in } E_n^c\} \subseteq
\biggl\{\lim_{t\uparrow\infty}\frac{1}{t}\int_0^t\indic_{\{X_s\in
E_n\}} \,d s=0\biggr\}.
\]
Therefore, if~$X$ is transient it follows that $\bprob\in\bPi$.
\end{pf*}

Another result giving a condition on whether $\prob^* \in\Pi^*$ based
on the
tail-decay of the distribution of $\zeta$ under $\qprob$ will be established.
%
%
\begin{prop}\label{PQtailsprop1}
Let Assumption~\ref{assbasic} hold. If $\bprob[\zeta<\infty
] = 0$ and
%
%
\begin{equation}\label{EQtails1}
\liminf_{t \uparrow\infty} \biggl(- \frac{1}{t} \log\qprob
[\zeta> t ]\biggr) \geq\blambda,\vspace*{-1pt}
\end{equation}
then $\bprob\in\bPi$.\vspace*{-2pt}
\end{prop}

When $\blambda= 0$ the fact that $\qprob[\zeta> t]\leq1$
immediately yields that $\bprob\in\bPi$.\vspace*{-2pt}
%
%
\begin{cor}\label{CQtailsprop1}
Let Assumption~\ref{assbasic} hold. If $\blambda= 0$ and $\bprob
[\zeta< \infty] = 0$ then $\bprob\in\bPi$.
\end{cor}
\begin{pf*}{Proof of Proposition~\ref{PQtailsprop1}}
That $\bprob\in\Pi$ follows by Proposition
\ref{Prelarbrobgrowthprop}. Also, by Proposition \ref
{Prelarbrobgrowthprop}, using the fact that $V^*_t = \exp(\lambda^* t)
\eta^* (X_t)$ for $t \in\Real_+$,
\[
\log\biggl(\mathbb{E}^{\bprob}\biggl[\frac{1}{\eta
^*(X_t)}\biggr]\biggr) = \lambda^*
t +
\log(\qprob[\zeta> t]) - \log\eta^* (x_0).\vspace*{-1pt}
\]
Thus, (\ref{EQtails1}) implies
%
%
\begin{equation}\label{Elogexpecval}
\limsup_{t \uparrow\infty} \biggl(\frac{1}{t} \log\biggl(\mathbb
{E}^{\prob^*}\biggl[\frac{1}{\eta^*(X_t)}\biggr]\biggr)
\biggr)\leq0.\vspace*{-1pt}
\end{equation}
Now, by Chebyshev's inequality, for each $\epsilon> 0$,
\begin{eqnarray*}
\frac{1}{t} \log\biggl(\prob^* \biggl[\frac{1}{t} \log\breta(X_t)
\leq- \epsilon\biggr]\biggr) &=& \frac{1}{t} \log\biggl(\prob^*
\biggl[\frac{1}{\breta(X_t)} \geq\exp(\epsilon t
)\biggr]\biggr)
\\[-2pt]
&\leq&\frac{1}{t} \log\biggl(\exp(-\epsilon t)\mathbb
{E}^{\bprob
}\biggl[\frac{1}{\breta(X_t)}\biggr]\biggr) \\[-2pt]
&=& - \epsilon+ \frac{1}{t} \log\biggl(\mathbb{E}^{\bprob}
\biggl[\frac{1}{\breta(X_t)}\biggr] \biggr).\vspace*{-1pt}
\end{eqnarray*}
In conjunction with (\ref{Elogexpecval}), this gives
\[
\limsup_{t \uparrow\infty} \biggl(\frac{1}{t} \log\biggl(\bprob
\biggl[\frac{1}{t} \log\eta^* (X_t) \leq- \epsilon\biggr]\biggr)
\biggr) \leq
- \epsilon,\vspace*{-1pt}
\]
which implies, in particular, that
\[
\lim_{t \uparrow\infty} \bprob\biggl[\frac{1}{t} \log\breta(X_t)
\leq- \epsilon\biggr] = 0.\vspace*{-1pt}
\]
Since this is true for all $\epsilon> 0$, it follows that $\prob^*
\in
\Pi^*$.\vspace*{-2pt}
\end{pf*}
%
%
\begin{rem}
From~\cite{MR1326606}, Theorem 4.4.4 (note that there, $\lambda_c$ is
used in place of $-\blambda$),
\[
-\blambda=
\lim_{n\uparrow\infty}\lim_{t\uparrow\infty}\frac{1}{t}\log
\qprob[\zeta_n > t].\vspace*{-1pt}
\]
Since $\qprob[\zeta_n>t]\leq\qprob[\zeta> t]$ it holds that
\[
\blambda+ \liminf_{t\uparrow\infty}\frac{1}{t}\log\qprob[\zeta> t]
\geq0.\vspace*{-1pt}\vadjust{\goodbreak}
\]
In particular, (\ref{EQtails1}) is really equivalent to
\[
\lim_{t \uparrow\infty} \biggl(\frac{1}{t} \log\qprob
[\zeta> t ]\biggr) = - \blambda.
\]
\end{rem}

\section{Connections with optimal arbitrages}\label{Srelarb}

In~\cite{KarFern10a}, and quite close to the setting considered here, the
authors treat the problem of \textit{optimal arbitrage} on a given
finite time
horizon. We briefly mention the main points below, sending the interested
reader to~\cite{KarFern10a} for a more in-depth treatment.

Consider a class of probabilities $(\prob_x)_{x \in E}$ on $(\Omega,
\F_\infty)$ under which the coordinate process~$X$ has Markovian
structure, and with the property that $\prob_x \llloc\qprob_x$ holds
for all $x \in E$. Define a function $U\dvtx\Real_+ \times E \mapsto[0,1]$
via the following recipe:
for $(T,x) \in\Real_+ \times E$, set
\[
1 / U(T,x) = \sup\{v \in\Real_+ \such\exists V \in\V\mbox{
such that } \prob_x[V_T \geq v] = 1\}.
\]
In words, $1 / U(T,x)$ is the maximal capital that one can realize at
time $T$
starting from unit initial capital when
the market configuration at the initial time is $x \in E$. Equivalently,
$U(T,x)$ is the minimal capital required in order to ensure at least
one unit
of wealth at time $T$ 
when the market
configuration at the initial time is $x \in E$. Arbitrage on the finite time
interval $[0, T]$ exists if and only if $U(T, x) < 1$. Using the
notation of
the present paper and recalling that for $x_0 \in E$ the subscripts in the
probability measures are dropped, it is shown in~\cite{KarFern10a} that
arbitrage over a time horizon $[0,T]$ exists if and only if $\qprob
[\zeta>
T] < 1$. Furthermore, it is established that $U(T, x) = \qprob_x
[\zeta
> T]$
for all $(T, x) \in\Real_+ \times E$, and that the optimal arbitrage exists
and is given by $V^T = (V^T_t)_{t \in[0,T]}$, where
%
%
\begin{equation} \label{eqoptarb}
V^{T}_t = \frac{\qprob[\zeta> T \such\F_t]}{\qprob[\zeta> T]} =
\frac
{U(T - t, X_t)}{U(T, x_0)}\qquad \mbox{for } t \in[0, T].
\end{equation}
Observe that the optimal arbitrage $V^T$ in (\ref{eqoptarb}) is
normalized so that $V^T_0 = 1$. In~\cite{KarFern10a}, the normalization
is such that the \textit{terminal} value of the optimal relative
arbitrage is unit; as already mentioned, in that case $U(T, x_0)$ is
the minimal capital required at time zero to ensure a unit of capital
at time $T$.
%
%
\begin{rem}
In~\cite{KarFern10a}, Sections 10--12, the problem of optimal
arbitrage is specified to when $E$ is the interior of the simplex on
$\reals^{d-1}$, that is,
\[
E = \Biggl\{x \in\reals^{d-1}  \bigm|\min_{i=1,\ldots, d - 1}x_i >
0 \mbox{, and }  \sum_{i=1}^{d-1}x_i < 1\Biggr\}.
\]
(In fact, in~\cite{KarFern10a} the simplex \mbox{$\Delta^{d}_{+} :=
\{x \in\reals^{d}  |\min_{i=1,\ldots, d}x_i > 0\mbox{,  and }
\sum_{i=1}^{d}x_i = 1\}$} is used. Since $x =
(x_i)_{i=1, \ldots,
d-1} \in E \Longleftrightarrow(x,1-\sum_{i=1}^{d-1}x_i)
\in
\Delta^{d}_{+}$, $E$ is in trivial one-to-one correspondence with
$\Delta^{d}_{+}$. For the purposes of this paper, the state space has
to be an open set; for this reason, $E$ as defined above will be used
throughout.) The interpretation is that the coordinate process~$X$
represents the
\textit{relative} capitalizations of stocks, and the corresponding optimal
arbitrages are in fact \textit{relative} arbitrages with respect to the market
portfolio. In principle, the treatment of~\cite{KarFern10a} does not really
utilize the special structure of the simplex; therefore, the general
case is
considered.
\end{rem}

It is natural to study the asymptotic behavior of these optimal
arbitrages as
the time-horizon becomes arbitrarily large. It is shown below that, under
suitable assumptions, the sequence of wealth processes $(V^T)_{T \in
\Real_+}$
(parameterized via their maturity) converges to the robust asymptotically
growth-optimal wealth process.

A tool in proving this convergence will be Proposition \ref
{Prelarbrobgrowthprop}. In view of that result, it follows that if
$\blambda> 0$ and $\bprob_{x}[\zeta<\infty] = 0$ for each
$x\in
E$, arbitrage occurs if and only if the local $\bprob_{x}$-martingale
$1 / V^*$ is a strict local $\bprob_{x}$-martingale in the terminology
of~\cite{MR1478722}. If
$1 / V^*$ is a $\bprob_{x}$-martingale, then, even
though arbitrage does not exist, it is still possible to construct
robust growth-optimal trading strategies, as seen in Example
\ref{Exacorrgbmrelcap}.
%
%
\begin{rem}
Equation (\ref{Eqprobexitrep}) holds when $\zeta,\breta,\bV$ and
$\bprob_x$ are replaced by $\zeta_n,\breta_n,\bV_n$ and $\bprob_{x,n}$,
where these quantities appear in the
proof of Theorem~\ref{thmasymptogrowth1}. In
this case, and when $E=(0,\infty)^d$, conditioning upon $\zeta_n > T$
can be interpreted as forcing a diversity condition in
the market since $X\in E_n$ implies there exists some $\delta> 0$ such
that no
one asset's relative capitalization is above $1-\delta$. Conditioned upon
never exiting $E_n$ for $\nin$, the robust growth optimal wealth
process $\bV_n$ is thus
identified with the long-run version of the arbitrage constructed in
\cite{OstRhe}.
\end{rem}

Equation (\ref{Eqprobexitrep}) may be re-written as
%
%
\begin{equation}\label{Eqprobexitrepalt}
e^{\blambda T}\qprob_x[\zeta> T] =
\breta(x)\mathbb{E}^{\bprob}_{x}\biggl[\frac{1}{\breta(X_T)}\biggr].
\end{equation}
Thus, to study the asymptotic behavior of $V^T_t$ as $T\uparrow\infty
$ in
(\ref{eqoptarb}), it is
necessary to study the long-time (as $T\uparrow\infty$) behavior of
$\mathbb{E}^{\bprob}_{x}[(\breta(X_T))^{-1}]$. Assume
that~$X$ is positive
recurrent (or, equivalently, tight) under $(\bprob_x)_{x\in E}$ with
invariant probability measure $\mu$. Under Assumption~\ref{assbasic},
\cite{MR1152459}, Theorem~1.2~(iii), equations
(3.29), (3.30) extends the ergodic result in
(\ref{Eposrecergodicthm}) to functions~$f$ which are integrable with
respect to $\mu$. Thus, for all positive measurable functions
$f\dvtx E\mapsto\reals$,
%
%
\begin{equation}\label{Eposrecergodicthmposf}
\lim_{T\uparrow\infty}\mathbb{E}^{\bprob}_{x}[f(X_T)]=
\int_Ef\,d\mu,
\end{equation}
and this limit is the same for all $x\in E$. This yields the
following proposition:

%
\begin{prop} \label{PQtailsprop2}
Let Assumption~\ref{assbasic} hold. Suppose that $\breta\in
H_{\blambda
}$ is
such that
%
%
\begin{equation}\label{Ebretato0}
\lim_{n\uparrow\infty}\sup_{x\in E_n^c}\breta(x) = 0.
\end{equation}
Then $\bprob_x[\zeta< \infty] = 0$ for all $x \in E$, and the
following are equivalent:
\begin{longlist}[(3)]
\item[(1)] $\lim_{T\uparrow\infty}e^{\blambda T}\qprob_x
[\zeta> T] = \kappa
\breta(x)$ for all $x\in E$ where $\kappa> 0$ does not depend upon~$X$;
\item[(2)] $\limsup_{T\uparrow\infty} e^{\blambda T}\qprob_x
[\zeta> T] <
\infty$ for some $x\in E$;
\item[(3)]~$X$ is positive recurrent under $(\bprob_x
)_{x\in E}$ and
$\int_E(\breta)^{-1}\,d\mu< \infty$ where~$\mu$ is the invariant measure
for~$X$.
\end{longlist}
\end{prop}
%
%
\begin{rem}
Note that $(3)$ implies $(1)$
even if (\ref{Ebretato0}) does not hold. Note
also that, by Example~\ref{Exaqexpllam0} below, some condition like
(\ref{Ebretato0}) is necessary for $(1),(2)$ and $(3)$ to be equivalent.
\end{rem}
\begin{pf*}{Proof of Proposition~\ref{PQtailsprop2}}
Let $x\in E$. Note that $(e^{\blambda t}\breta(X_t))^{-1}$
is a~$\bprob_x$ super-martingale. By (\ref{Ebretato0}), if $\bprob_x
[\zeta< \infty] > 0$, then the super-martingale property would
be violated.
Thus an
explosion cannot occur.

Regarding the equivalences, $(1)\Rightarrow(2)$ is trivial. As for
$(2)\Rightarrow(3)$, if $(2)$ holds, then by (\ref{Eqprobexitrepalt}) it
follows that there is some $T_0\geq0$ such that
\[
\sup_{T\geq T_0} \mathbb{E}^{\bprob}_{x}\biggl[\frac{1}{\breta
(X_T)}\biggr] <
\infty.
\]
Therefore, (\ref{Ebretato0}) yields that
$(X_T)_{T\geq T_0}$ form a $\bprob_x$ tight family of random
variables for each $x\in E$. By Proposition~\ref{Ptighteqposrec} it
follows that~$X$ is positive recurrent under $(\bprob_x
)_{x\in
E}$; hence, (\ref{Eposrecergodicthmposf}) gives
\[
\int_E\frac{1}{\breta}\,d\mu=
\lim_{T\uparrow\infty}\mathbb{E}^{\bprob}_{x}\biggl[\frac
{1}{\breta(X_T)}\biggr]
\leq
\limsup_{T\uparrow\infty}\mathbb{E}^{\bprob}_{x}\biggl[\frac
{1}{\breta(X_T)}\biggr] <
\infty
\]
proving $(3)$. Implication $(3)\Rightarrow(1)$ follows by applying
(\ref{Eposrecergodicthmposf}) to $1/\breta$ and using
(\ref{Eqprobexitrepalt}).
\end{pf*}

The following is the main result of the section.
%
%
\begin{theorem} \label{Tconvtooptarb}
Suppose that $\breta\in H_{\blambda}$ is such that $\bprob
[\zeta< \infty]
= 0$ and that condition $(1)$ in Proposition~\ref{PQtailsprop2} holds. Fix
$\prob\in\Pi$. Then, for any fixed $t \in\Real_+$,
%
%
\begin{equation} \label{eqconvtooptarb1}
\plim_{T \to\infty} \sup_{\tau\in[0, t]} \vert V^T_\tau-
V^*_\tau\vert= 0.
\end{equation}
Additionally, for each $T \in\Real_+$, let $(\vartheta^T_t)_{t \in
[0, T]}$
be a predictable process such that
%
%
\begin{equation}\label{EVrelarbrep}
V^T = 1 + \int_{0}^\cdot V^T_t (\vartheta_t^T)' \,\ud X_t.\vadjust{\goodbreak}
\end{equation}
With $\bell$ as in (\ref{Ebelldef}) and $\vartheta^* = \nabla\bell
(X)$, it follows that, for any fixed $t\in\Real_{+}$,
%
%
\begin{equation} \label{eqconvtooptarb2}
\plim_{T \to\infty} \int_0^t (\vartheta^T_\tau- \vartheta
^*_\tau)'
c(X_\tau) (\vartheta^T_\tau- \vartheta^*_\tau) \,\ud
\tau= 0.
\end{equation}
\end{theorem}
\begin{pf}
Fix $t \in\Real_+$. Equation (\ref{eqoptarb}), coupled with
condition (1)
in Proposition
\ref{PQtailsprop2}, implies that $\plim_{T \to\infty} V^T_t =
V^*_t$. Let $Z^T = (Z^T_\tau)_{\tau\in[0, t]}$ be defined via
$Z^T_\tau:= V^T / V^*$. The arguments used in the proof of Theorem
\ref{thmasymptogrowth1} show that $V^*$ is the num\'eraire portfolio in $\V$
under $\prob^*$, that is, that $Z^T$ is a~nonnegative $\prob
^*$-supermartingale on $[0,t]$ for
all $T \in(t, \infty)$. Then~\cite{Kard10}, Theorem 2.5, implies that
$\prob^*$-$\lim_{T \to\infty} \sup_{\tau\in[0, t]} \vert
Z^T_\tau- 1\vert=
0$. Using the fact that $\prob^* [\inf_{\tau\in[0, t]}
V^*_\tau> 0] =1$,
it follows that $\prob^*$-$\lim_{T \to\infty} \sup_{\tau\in[0,t]}
\vert V^T_\tau- V^*_\tau\vert= 0$. Now, with $R^T =
(R^T_\tau)_{\tau\in[0,t]}$ defined via
\[
R^T = \int_0^\cdot(\vartheta^T_s - \vartheta^*_s)'
\bigl(\ud X_s - c(X_s) \nabla\ell^* (X_s) \,\ud s\bigr),
\]
it holds that $Z^T = 1 + \int_0^\cdot Z^T_s \,\ud R_s$. Invoking
\cite{Kard10}, Theorem 2.5, again yields $\prob^*$-$\lim_{T \to\infty}
[R^T, R^T]_t = 0$ for all $t \in\Real_+$. As
\[
[R^T, R^T]_t = \int_0^t (\vartheta^T_\tau- \vartheta^*_\tau
)'
c(X_\tau) (\vartheta^T_\tau- \vartheta^*_\tau) \,\ud
\tau,
\]
(\ref{eqconvtooptarb2}) follows, with $\prob^*$ replacing
$\prob$ there.

Up to now, the validity of both (\ref{eqconvtooptarb1})
and (\ref{eqconvtooptarb2}), for the special case $\prob= \prob^*
\in
\Pi$ has been shown. For a general $\prob\in\Pi$, the result follows
by noting that~$\prob^*$ and
$\prob$ are equivalent on each $\F_{\zeta_n}$, $\nin$, and that
$\limn
\prob
[\zeta_n > t] = 1$. Indeed, for any $\epsilon> 0$ pick
$n_\epsilon
\in\Natural$ large enough so that $\prob
[\zeta_{n_{\epsilon}} \leq t] \leq\epsilon/ 2$. Then pick
$\delta
_\epsilon> 0$ so that $
\prob[A] \leq\epsilon/2$ holds whenever $A \in\F_{\zeta
_{n_{\epsilon
}}}$ and $\prob^*[A] \leq\delta_\epsilon$. Finally, pick
$T_\epsilon
\in\Real_+$ large enough so that
\[
\prob^* \Bigl[{\sup_{\tau\in[0, t]}} \vert V^T_\tau-
V^*_\tau\vert\geq\epsilon\Bigr] \leq\delta_\epsilon
\]
as well as
\[
\prob^* \biggl[\int_0^t (\vartheta
^T_\tau- \vartheta^*_\tau)' c(X_\tau) (\vartheta
^T_\tau- \vartheta^*_\tau) \,d \tau\geq\epsilon\biggr] \leq
\delta
_\epsilon
\]
holds whenever $T \geq T_\epsilon$. Therefore, for all $T \geq
T_\epsilon$,
\begin{eqnarray*}
\prob\Bigl[{\sup_{\tau\in[0, t]}} \vert V^T_\tau- V^*_\tau
\vert\geq\epsilon\Bigr] &\leq& \prob\Bigl[{\sup_{\tau\in
[0, \zeta_n \wedge t]}} \vert V^T_\tau- V^*_\tau\vert
\geq\epsilon\Bigr] + \prob
[\zeta_{n_{\epsilon}} \leq t] \\
&\leq& \epsilon/ 2 + \epsilon/2  =  \epsilon.
\end{eqnarray*}
This establishes (\ref{eqconvtooptarb1}). Similarly, we establish
(\ref{eqconvtooptarb2}).
\end{pf}
%
%
\begin{rem}
The result of Theorem~\ref{Tconvtooptarb} is expected to hold in
greater generality than its assumptions suggest. It is conjectured that the
results hold under Assumption~\ref{assbasic},\vadjust{\goodbreak} but it is an open
question. See Example~\ref{Exanotinteg} in Section~\ref{Sexamples}
for a
potential counterexample. The next example shows that it can even hold when
$\lambda^* = 0$.
\end{rem}
%
%
\begin{exa}\label{Exaqexpllam0}
Let $E = (0, \infty)$ and $c (x) = 1$ for $x \in E$. It is
straightforward to check that
\[
U(T, x) = \qprob_x [\zeta> T] = 2 \Phi\bigl(x / \sqrt
{T}\bigr) - 1\qquad
\mbox{for } (T, x) \in\Real_+ \times E,
\]
where $\Phi$ is the cumulative distribution function of the standard normal
law. With $x_0 = 1$, it follows that
\[
V^T_t = \frac{2 \Phi(X_t / \sqrt{T - t}) - 1}{2 \Phi
(1 / \sqrt{T}) - 1}\qquad \mbox{for } t \in[0, T].
\]
From this explicit formula it is straightforward that $\plim_{T \to
\infty} \sup_{\tau\in[0, t]} \vert V^T_\tau- X_\tau\vert= 0$ holds
whenever $t \in \Real_+$. Observe that $V^* = X$ exactly for the choice
$\eta^* (x) = x$ corresponding to $\lambda^* = 0$, and $\prob^*$ being
the probability that makes~$X$ behave as a three-dimensional Bessel
process. Remember that in this example the dimensionality of the set of
principal eigenfunctions is two---the other one is $\eta\equiv1$. It is
interesting to note that the sequence $(V^T)$ ``chooses'' to converge
to the optimal strategy of the optimal probability $\prob^*$ that
satisfies $\prob^* \in\Pi$.
\end{exa}

As in~\cite{KarFern10b}, Section 5.1,
for $T\in\Real_{+}$ and $x\in E$, define the measure $\prob^{\star
,T}_x$ on~$\F_T$ via
\[
\prob^{\star,T}_x [A] = \qprob_x[A \such\zeta> T]\qquad
\mbox{for } A\in\F_T.
\]
It is shown therein that, for each $t \in[0, T]$ and $x\in E$,
\[
\frac{d\prob^{\star,T}_x}{d\qprob_x}\bigg|_{\F_t} =
\frac{U(T-t,X_t)}{U(T,x)}\indic_{\{\zeta> t\}}.
\]
Furthermore, under the assumption $U\in C^{1,2}((0,T)\times E)$, the
coordinate process~$X$ under
$( \prob^{\star,T}_x )_{x\in E}$ has dynamics on $[0,T]$ of
\begin{eqnarray*}
\ud X_{\tau} &=&
c(X_\tau)\frac{\nabla_{x}U(T-\tau,X_\tau)}{U(T-\tau,X_{\tau
})}\,d\tau+
\sigma(X_{\tau})\,\ud W^{\prob^{\star,T}}_{\tau}\\
&=&c(X_\tau)\vartheta^T_\tau \,d\tau+
\sigma(X_{\tau})\,\ud W^{\prob^{\star,T}}_{\tau}
\end{eqnarray*}
using the notation of (\ref{EVrelarbrep}) in Theorem \ref
{Tconvtooptarb}. Assuming
$\bprob_x[\zeta< \infty] = 0$, it follows that $\prob^{\star,T}_x$ and
$\bprob_x$ are equivalent on $\F_t$ for $t \in[0,T]$ with
%
%
\begin{eqnarray}\label{Ebprobpstarrnderiv}
\frac{d\prob^{\star,T}_x}{d\bprob_x}\bigg|_{\F_t} &=&
\exp\biggl( - \frac{1}{2} \int_0^t (\vartheta^T_\tau-
\vartheta^{*}_\tau)' c(X_\tau) (\vartheta^T_\tau-
\vartheta^{*}_\tau) \,d \tau\nonumber\\[-8pt]\\[-8pt]
&&\hspace*{59pt}{} + \int_0^t(\vartheta^T_\tau-
\vartheta^{*}_\tau)'\sigma(X_\tau) \,\ud
W^{\bprob}_\tau\biggr).\nonumber
\end{eqnarray}
Thus the results of Theorem~\ref{Tconvtooptarb} immediately imply the
following:
%
%
\begin{prop}\label{PTcondvarnormconvergence}
Suppose the hypotheses of Theorem~\ref{Tconvtooptarb} hold. Then, for any
$t\in\Real_{+}$, $\prob^{\star,T}_x$ converges in variation
norm to $\bprob_x$ on $\F_t$ as $T\uparrow\infty$.
\end{prop}
\begin{pf}
The process on the right-hand side of (\ref{Ebprobpstarrnderiv}) is the
process $Z^T = V^T / V^*$ in the proof of Theorem~\ref{Tconvtooptarb}.
Since, for
each $A\in\F_t$,
\[
|\prob^{\star,T}_x(A) - \bprob_x(A)| \leq
\mathbb{E}^{\bprob}_{x}[|Z^T_t-1|],
\]
the result follows from~\cite{Kard10}, Theorem 2.5(i).
\end{pf}
%
%
\begin{rem}
In~\cite{MR781410}, a similar result to Proposition
\ref{PTcondvarnormconvergence} is obtained, though not in the setting of
convergence of relative arbitrages. Namely, it is assumed that
%
%
\begin{equation}\label{Egradlogconvcond}
\lim_{T\uparrow\infty}\frac{\nabla_{x}U(T,x)}{U(T,x)} = \nabla
\bell(x)\qquad
\mbox{for } x \in E,
\end{equation}
where the convergence takes place exponentially fast with rate
$\blambda
$ and
is uniform on compact subsets of $E$. Under this assumption, the measures~$\prob^{\star,T}_x$ are shown to weakly converge as $T \uparrow
\infty$
to $\bprob_x$ on $\F_t$ for
each $t\in\Real_+$.

In the case where $E$ is bounded with smooth boundary, and $c$ is uniformly
elliptic over $E$, (\ref{Egradlogconvcond}) holds if there exists a
function $H\dvtx E\mapsto\reals$ such that, for each $i=1,\ldots, d$,
\[
\sum_{j=1}^{d}c_{ij}(x)\,\frac{\partial}{\partial_{x_j}}H(x) =
f_{i}(x);\qquad  f_i(x):= -\frac{1}{2}\sum_{j=1}^{d}\frac{\partial
}{\partial_{x_j}}c_{ij}(x),\qquad i=1,\ldots,d.
\]
In vector notation, this gradient condition takes the form $\nabla H =
c^{-1}f$, and~$f$ is the \textit{Fichera drift} associated to
$\qprob$. Under
this hypothesis, the measure $m(dx) = \exp(2H(x))\,dx$ is
reversing for the
transition probability function $\qprob(t,x,\cdot)$, and the
convergence result
in (\ref{Egradlogconvcond}) follows by representing $U(T,x) =
\qprob_x [\zeta> T]$ as an eigenfunction expansion where the
underlying space
is $L^2(E,m)$; see~\cite{MR781410}.
\end{rem}

\section{A thorough treatment of the one-dimensional case}
\label{Sone-dim-case}

This section considers the case $d=1$, where $E=(\alpha,\beta)$ is a
bounded interval. If $E=\reals$, then $\blambda= 0$ holds by
Proposition~\ref{Pblambdatest}, because the coordinate process under
$\qprob$ is recurrent. If $E$ is a half-bounded interval, it is
possible for:
\begin{itemize}
\item$\blambda= 0$, even though there is explosion under $\qprob$;
see Example~\ref{Exaqexpllam0}.
\item$\blambda> 0$, even though there is no explosion under $\qprob$;
see Example~\ref{Exacorrgbm} with $d=1$.
\end{itemize}
Hence making a general statement connecting $\blambda> 0$ with explosion
or nonexplosion under $\qprob$ is difficult. Thus to enlighten the
connections with relative arbitrages, the following will assumed
throughout the
section:\vadjust{\goodbreak}
%
%
\begin{ass}\label{Aonedim}
Assumption~\ref{assbasic} holds for $E=(\alpha,\beta)$ with $-\infty<
\alpha< \beta< \infty$.
\end{ass}

Under the validity of Assumption~\ref{Aonedim}, results are provided
that almost completely cover all the cases that can occur.

The first proposition establishes point-wise tests for $c$ which
yield $\blambda> 0$ or $\blambda= 0$. The second
proposition gives integral tests which yield $\blambda>
0$ or \mbox{$\blambda= 0$}. Condition (\ref{Eintblambdapossl}) is equivalent to
the coordinate process~$X$ under $(\qprob_x)_{x\in
[\alpha
,\beta]}$, exploding
to both $\alpha,\beta$ with positive probability. Additionally,
condition (\ref{Eintblambdapossl}) not
only yields $\blambda> 0$ but also that $\bprob\in\bPi_{\mathrm{a.s.}}$ (and hence
$\bprob\in\bPi$).

Recall the following facts regarding explosion, transience,
recurrence and positive recurrence in the one-dimensional case under
Assumption~\ref{Aonedim}; see~\cite{MR1326606}, Chapter 5.1:
\begin{itemize}
\item Since $E$ is bounded the
coordinate process~$X$ under $(\qprob_x)_{x\in[\alpha
,\beta
]}$ is
transient. Furthermore it explodes to $\alpha$ and/or $\beta$ with positive
probability if, for some $x_0\in(\alpha,\beta)$,
\[
\int_\alpha^{x_0}\frac{x-\alpha}{c(x)}\,dx<\infty
\quad\mbox{and/or}\quad
\int_{x_0}^\beta\frac{\beta-x}{c(x)}\,dx<\infty.
\]
\item The coordinate process~$X$ under $(\bprob_x)_{x\in
(\alpha,\beta)}$ is recurrent if
%
%
\begin{equation}\label{Ebprobrecurrent}
\int_\alpha^{x_0} \frac{1}{(\breta(x))^2}\,dx = \infty
\quad\mbox{and}\quad
\int_{x_0}^\beta\frac{1}{(\breta(x))^2}\,dx = \infty.
\end{equation}
If either of the integrals in (\ref{Ebprobrecurrent}) are finite, then the
coordinate process~$X$ is transient towards the endpoint with finite
integral.
\item The coordinate process~$X$ under
$(\bprob_x)_{x\in(\alpha,\beta)}$ is positive
recurrent if
(\ref{Ebprobrecurrent}) holds and if
%
%
\begin{equation}\label{Ebprobposrecurrent}
\int_\alpha^\beta\frac{(\breta(x))^2}{c(x)}\,dx < \infty.
\end{equation}
\end{itemize}
%
%
\begin{prop}[(Pointwise result)]\label{Ppwprop}
Let Assumption~\ref{Aonedim} hold. If
%
%
\begin{equation}\label{Epwblambdapos}
\sup_{x\in(\alpha,\beta)}\frac{(x-\alpha)^2(\beta-x)^2}{c(x)}<
\infty,
\end{equation}
then $\blambda> 0$. If
%
%
\begin{equation}\label{Epwblambdazero}
\lim_{x\downarrow\alpha}\frac{(x-\alpha)^2}{c(x)} = \infty
\quad\mbox{or}\quad \lim_{x\uparrow
\beta}\frac{(\beta-x)^2}{c(x)} = \infty,
\end{equation}
then $\blambda= 0$.
\end{prop}
%
%
\begin{rem}
We thank an anonymous referee for suggesting the short, self-contained proof
to Proposition~\ref{Ppwprop} below.\vadjust{\goodbreak}
\end{rem}
\begin{pf*}{Proof of Proposition~\ref{Ppwprop}}
By~\cite{MR1326606}, Theorem 4.4.5
(note that $\lambda_c$ from~\cite{MR1326606},
Theorem 4.4.5, is equal to $-\blambda$ here), $\blambda$ admits the
following variational representation:
%
%
\begin{equation}\label{Eblambdapw01rep}
\blambda= \sup_{\stackrel{\eta\in C^2(\alpha,\beta)}{\eta>
0}}\inf
_{x\in
(\alpha,\beta)} \frac{-c(x)\eta''(x)}{2\eta(x)},
\end{equation}
where the $'$ symbol is used to signify a derivative with respect to
$x$ (and
not to denote matrix transposition as it was used in previous sections).

Let $\eta(x) = \sqrt{(x-\alpha)(\beta-x)}$.
If (\ref{Epwblambdapos}) holds, then
\[
\inf_{x\in
(\alpha,\beta)} \frac{-c(x)\eta''(x)}{2\eta(x)} = \inf_{x\in
(\alpha
,\beta)}
\frac{(\beta-\alpha)^2c(x)}{8(x-\alpha)^2(\beta-x)^2} > 0
\]
and hence $\blambda> 0$.

Now, assume (\ref{Epwblambdazero}) holds for $x\downarrow
\alpha$. The proof for $x\uparrow\beta$ is the same. Let \mbox{$a>\alpha$},
and consider the case when
$c\equiv1$ and $E=(\alpha,a)$. Since Assumption~\ref{assbasic} clearly
holds in this
setting, let $\blambda_a$ represent the generalized
principle eigenvalue. Set $\lambda_a = \frac{\pi^2}{2(a-\alpha)^2}$ and
consider the function $\phi(x) =
\sin(\sqrt{2\lambda_a}(x-\alpha))$. It can be directly verified
that $-\frac{1}{2}\phi''(x) = \lambda_a\phi(x)$ and that
both~(\ref{Ebprobrecurrent}) and~(\ref{Ebprobposrecurrent}) hold [with $c\equiv1$, $\beta$ replaced
by $a$
and $x_0\in(\alpha,a)$]. Thus, Proposition~\ref{Pblambdatest} implies
that $\blambda_a = \lambda_a = \frac{\pi^2}{2(a-\alpha)^2}$. Plugging
this into~(\ref{Eblambdapw01rep}) (again, for $c\equiv1$ and~$\beta$ replaced by $a$) gives for all $\eta\in C^2(\alpha,a), \eta
> 0$
%
%
\begin{equation}\label{Eetaabpiresult}
\inf_{x\in(\alpha,a)}\frac{-\eta''(x)}{2\eta(x)}\leq\frac{\pi
^2}{2(a-\alpha)^2}.
\end{equation}

Now, for the general case, it is clearly true that $\blambda\geq
0$. Assume by way of contradiction that $\blambda> 0$. By
(\ref{Eblambdapw01rep}) it follows that there exists a $\tilde{\lambda}>0$
and $\eta\in C^2(\alpha,\beta), \eta> 0$ such that
%
%
\begin{equation}\label{Ecpwprooftempval}
\tilde{\lambda}\leq\inf_{x\in(\alpha,\beta)}\frac{-c(x)\eta
''(x)}{2\eta(x)}.
\end{equation}
Let $M>0$. Since (\ref{Epwblambdazero}) holds, there is an $\alpha_M$ such
that for $x\in(\alpha,\alpha_M)$,
%
%
\begin{equation}\label{EtempMbound}
M\leq\frac{(x-\alpha)^2}{c(x)}\leq\frac{(\alpha_M-\alpha)^2}{c(x)}.
\end{equation}
Together, (\ref{Ecpwprooftempval}) and (\ref{EtempMbound}) give
%
%
\begin{equation}\label{Ecpwprooftempval2}
\frac{\tilde{\lambda}M}{(\alpha_M-\alpha)^2}\leq\inf_{x\in
(\alpha
,\alpha_M)}\frac{\tilde{\lambda}}{c(x)}\leq\inf_{x\in(\alpha
,\alpha
_M)}\frac{-\eta''(x)}{2\eta(x)}.
\end{equation}
By (\ref{Eetaabpiresult}) with $a=\alpha_M$, it follows that
%
%
\begin{equation}\label{Ecpwprooftempval3}
\inf_{x\in(\alpha,\alpha_M)}\frac{-\eta''(x)}{2\eta(x)}\leq
\frac{\pi^2}{2(\alpha_M-\alpha)^2}.
\end{equation}
Combining (\ref{Ecpwprooftempval2}) and (\ref{Ecpwprooftempval3}) gives
\[
\frac{\tilde{\lambda}M}{(\alpha_M-\alpha)^2} \leq
\frac{\pi^2}{2(\alpha_M - \alpha)^2}
\]
or that $M\leq\pi^2 /(2\tilde{\lambda})$. This is a contradiction
since $M$ was
arbitrary. Thus $\blambda= 0$.
\end{pf*}

The proof of the following result is lengthy and technical; for this
reason, it is delayed until Section~\ref{Sonedimproofs}.
%
%
\begin{prop}[(Integral result)]\label{Pintprop}
Let Assumption~\ref{Aonedim} hold. If
%
%
\begin{equation}\label{Eintblambdapossl}
\int_\alpha^\beta\frac{(x-\alpha)(\beta-x)}{c(x)}\,dx < \infty,
\end{equation}
then:
\begin{longlist}[(4)]
\item[(1)] $\blambda> 0$.
\item[(2)] For any $\breta\in H_{\blambda}$, $\lim_{x\downarrow
\alpha
}\breta(x) = 0 = \lim_{x\uparrow\beta}\breta(x)$.
\item[(3)] For any $\breta\in H_{\blambda}$, the coordinate process
$X$ under $(\bprob_x)_{x\in(\alpha,\beta)}$ is positive
recurrent and so by Proposition~\ref{Pblambdatest}, $\breta$ is unique
up to multiplication by a positive constant.
\item[(4)] $\bprob\in\bPi_{\mathrm{a.s.}}$ and hence $\bprob\in\bPi$.
\end{longlist}
If, for some $a\in
(\alpha,\beta)$,
%
%
\begin{equation}\label{Eintblambdazero}
\int_{\alpha}^{a}\frac{(x-\alpha)^2}{c(x)}\,dx = \infty
\quad\mbox{or}\quad \int_{a}^{\beta}\frac{(\beta-x)^2}{c(x)}\,dx
= \infty,
\end{equation}
then $\blambda= 0$.
\end{prop}

\section{Examples}\label{Sexamples}

\subsection{One-dimensional examples}\label{SSonedimex}

The following examples display a variety of outcomes regarding $\breta
$ and
$\bprob$. Proofs of all the statements follow from Propositions
\ref{Ppwprop},~\ref{Pintprop} and/or from the tests for recurrence, null
recurrence or positive recurrence under $\bprob$ given in equations
(\ref{Ebprobrecurrent}) and (\ref{Ebprobposrecurrent}) in conjunction
with Proposition~\ref{Pblambdatest}.

The first three Examples~\ref{Exa1}--\ref{Exa3}, all consider
$E=(0,1)$ and display the different possible outcomes depending upon
the rate of decay (to zero) of $c$ at $0$ and~$1$. The fourth Example
\ref{Exanotinteg} shows that it is possible that $\blambda> 0$,
$\bprob\in\bPi_{\mathrm{a.s.}}$, and the coordinate process is positive
recurrent under $\bprob$, while $(\breta)^{-1}$ fails to be integrable
with respect to the invariant measure under $\bprob$; thus, the results
of Section~\ref{Srelarb}, and in particular Theorem \ref
{Tconvtooptarb}, are not applicable. Finally, Example \ref
{Exaimpossible} shows that even if $\blambda> 0$, there is no explosion
of~$X$ under $\bprob$ and $\breta$ is unique (up to multiplication), no
conclusion can be made as to if $\bprob\in\bPi _{\mathrm{a.s.}}$ or
$\bprob\in\bPi$, based on results of this article.
%
%
\begin{exa}\label{Exa1}
Let $E=(0,1)$ and $c(x) = x(1-x)$. Then:

\begin{itemize}
\item Equation (\ref{Eintblambdapossl}) holds and so the results of Proposition
\ref{Pintprop} follow.
\item$\breta(x) = x(1-x)$, $\blambda= 1$.
\item Equation (\ref{Ebretato0}) holds as well as condition $(3)$ in Proposition
\ref{PQtailsprop2}. Thus the results of
Theorem~\ref{Tconvtooptarb} and Proposition \ref
{PTcondvarnormconvergence} follow.
\end{itemize}
\end{exa}
%
%
\begin{exa}\label{Exa2}
Let $E=(0,1)$ and $c(x) = x^2(1-x)^2$. Then:

\begin{itemize}
\item$\qprob[\zeta< \infty] = 0$.
\item$\breta(x) = \sqrt{x(1-x)}$, $\blambda= 1/8$.
\item The coordinate process~$X$ is null recurrent under
$(\bprob_x)_{x\in E}$; however, $\bprob\in\bPi_{\mathrm{a.s.}}$.
\end{itemize}

Note that there is a multidimensional generalization of this in Example~\ref{Exacorrgbmrelcap}.
\end{exa}
%
%
\begin{exa}\label{Exa3}
Let $E=(0,1)$ and $c(x) = x^3(1-x)^3$. Then:

\begin{itemize}
\item$\qprob[\zeta< \infty] = 0$.
\item$\blambda= 0$ by either Proposition~\ref{Ppwprop} or
\ref{Pintprop}.
\item$\breta$ can be any affine function $\alpha+ \beta x$ such that
$\breta
>0$ on $(0,1)$. For any such $\breta$, $\bprob\in\bPi_{\mathrm{a.s.}}$.
\end{itemize}
\end{exa}
%
%
\begin{exa}\label{Exanotinteg}
Let $E=(0,\hat{x})$, where
\[
\hat{x} :=\min\biggl\{x > 0 \Bigm|\int_0^x\log(-{\log}
(y))\,dy
=0\biggr\}\approx0.75.
\]
Furthermore, let $c \dvtx E \mapsto\Real_+$ be defined via
\[
c(x) = -2x\log(x)\int_0^x\log(-{\log}(y))\,dy\qquad \mbox
{for }
x \in E.
\]
Then:
\begin{itemize}
\item Equation (\ref{Eintblambdapossl}) holds and so the results of Proposition
\ref{Pintprop} follow.
\item$\breta(x) = \int_0^{x}\log(-\log(y))\,dy$,
$\blambda= 1$.
\item$(\breta)^{-1}$ is not integrable with respect to the
invariant measure for $\bprob$.
\end{itemize}
\end{exa}
%
%
\begin{exa}\label{Exaimpossible}
Let $E=(0,\infty)$ and
\[
c(x) = \frac{4(x^{3/2}\int_0^x\cos(y^{-1/2})\,dy +
4x^2 -
x^{5/2})}{2-\sin(x^{-1/2})}\qquad  \mbox{for } x
\in E.
\]
Then:
\begin{itemize}
\item$\qprob[\zeta< \infty] = 0$.
\item$\breta(x) = \int_0^x\cos(y^{-1/2})\,dy + 4\sqrt
{x} - x$,
$\blambda= 1$.
\item The coordinate process~$X$ under $(\bprob_x)_{x\in
E}$ is
null-recurrent. No conclusions as to whether or not
$\bprob\in\bPi_{\mathrm{a.s.}}$ or $\bPi$ can be drawn based on the results\vadjust{\goodbreak}
of the
paper (see Propositions~\ref{Phopfprop} and~\ref{Prelarbrobgrowthprop})
since
\begin{eqnarray*}
\limsup_{x\downarrow0}\biggl(\frac{1}{2}\nabla\bell(x)'c(x)\nabla
\bell
(x)-\blambda\biggr)
&=& 0, \\
\liminf_{x\downarrow0}\biggl(\frac{1}{2}\nabla\bell(x)'c(x)\nabla
\bell
(x)-\blambda\biggr)
&=& -\frac{2}{3}.
\end{eqnarray*}
\end{itemize}
\end{exa}

\subsection{Multi-dimensional examples}

The following examples show that the optimal $\breta$ need not vanish
on the
boundary of $E$ even when $E$ is bounded, and that strictly positive
asymptotic growth rate is
possible even when $\qprob[\zeta< \infty] = 0$.
%
%
\begin{exa}[(Correlated geometric Brownian motion)]\label{Exacorrgbm}
Let $E = (0,\infty)^{d}$, and define the matrix $c$ via
\[
c_{ij}(x) = x_i x_j A_{ij}, \qquad 1\leq i,j\leq d,
\]
where $A$ is a symmetric, strictly positive definite $d\times d$
matrix. Define the vectors $\hat{A}, \hat{B} \in\reals^d$ by
\[
\hat{A}_i = A_{ii}\qquad (1\leq i\leq d),\qquad  \hat{B} =
\tfrac{1}{2}A^{-1}\hat{A}.
\]
Then
%
%
\begin{equation}\label{Egbmbreta}
\breta(x) = \prod_{i=1}^{d} x_i^{\hat{B}_i},\qquad  \blambda=
\frac{1}{8}\hat{A}'A^{-1}\hat{A},
\end{equation}
and $\bprob\in\bPi_{\mathrm{a.s.}}$.

To see the validity of the above claims, set $\eta$, $\lambda$ as the
respective right-hand sides of (\ref{Egbmbreta}). A straightforward
calculation shows that
$L\eta= -\lambda\eta$ and hence that $\blambda\geq\lambda$. Set
$(\prob^{\eta}_{x})_{x\in\hE}$ as the solution to
the generalized martingale problem for~$L^{\eta}$, as in
(\ref{Edoobetatransform}) and $\prob^{\eta} =
\prob^{\eta}_{x_0}$. The coordinate process~$X$ under~$\prob^{\eta
}$ is
given by $X = \exp(a W)$ where $a$ is the unique positive
definite square
root of $A$ and $W$ a Brownian motion under $\prob^{\eta}$. Thus, under
$\prob^{\eta}$,
\[
\frac{1}{t}\log\eta(X_t) = \frac{1}{t}\hat{B}'aW_t.
\]
The strong law of large numbers for Brownian motion gives that $\prob
^{\eta}\in\bPi_{\mathrm{a.s.}}$. Theorem~\ref{thmasymptogrowth1} then yields
$\blambda\leq\sup_{V \in\V} g(V; \prob^{\eta}) \leq\lambda$,
and hence $\blambda=\lambda$, $\breta= \eta$ and $\bprob=\prob
^{\eta}$.
\end{exa}
%
%
\begin{exa}[(Relative capitalizations of a correlated geometric Brownian
motion)]\label{Exacorrgbmrelcap}
For $d \geq2$, let
\[
E = \Biggl\{x\in\reals^{d-1} \bigm| \min_{i=1,\ldots,d-1} x_i > 0;
\sum_{i=1}^{d-1}x_i < 1\Biggr\}.
\]
For the matrix $A$ of Example~\ref{Exacorrgbm}, define the
$(d-1)$-dimensional square matrix $\mathcal{A}$ by
\[
\mathcal{A}_{ij} = A_{ij} - A_{id} - A_{jd} + A_{dd},\qquad 1\leq
i,j\leq d-1,
\]
and the matrix $c$ via
\[
c_{ij}(x) = x_i
x_j\bigl(\mathcal{A}_{ij}-(\mathcal{A}x)_i -
(\mathcal{A}x)_j + x'\mathcal{A}x\bigr),\qquad
1\leq i,j\leq d-1.
\]
Set the $(d-1)$-dimensional vectors
\[
\hat{\mathcal{A}}_i = \mathcal{A}_{ii} \qquad(1\leq i\leq d-1),\qquad
\hat{\mathcal{B}} = \tfrac{1}{2}\mathcal{A}^{-1}\hat{\mathcal{A}}.
\]
Then
%
%
\begin{equation}\label{Egbmrelcapbreta}\quad
\breta(x)
=\Biggl(\prod_{i=1}^{d-1}x_i^{\hat{\mathcal{B}}_i}\Biggr)
\Biggl(1-\sum
_{i=1}^{d-1}x_i\Biggr)^{1-\sum_{i=1}^{d-1}\hat{\mathcal
{B}}_i},\qquad
\blambda=
\frac{1}{8}\hat{\mathcal{A}}'\mathcal{A}^{-1}\hat{\mathcal{A}},
\end{equation}
and $\bprob\in\bPi_{\mathrm{a.s.}}$. Furthermore, the
coordinate process under $\bprob$ on the simplex has the same
dynamics as the coordinate process under $\bprob$ in Example~\ref{Exacorrgbm}
moved to the simplex.

To prove the validity of the claims, set $\eta,\lambda$ as
the right-hand sides of (\ref{Egbmrelcapbreta}), that is,
\[
\eta(x)
=\Biggl(\prod_{i=1}^{d-1}x_i^{\hat{\mathcal{B}}_i}\Biggr)
\Biggl(1-\sum
_{i=1}^{d-1}x_i\Biggr)^{1-\sum_{i=1}^{d-1}\hat{\mathcal{B}}_i}
\qquad\mbox{for } x \in E,
\lambda=
\frac{1}{8}\hat{\mathcal{A}}'\mathcal{A}^{-1}\hat{\mathcal{A}}.
\]
A long calculation shows that $L\eta= -\lambda\eta$. Let
$(\prob^{\eta}_x)_{x\in\hE}$ be the solution to the
generalized martingale problem for $L^{\eta}$ as in (\ref
{Edoobetatransform}), and set $\prob^{\eta} = \prob^{\eta}_{x_0}$.

Rewrite $\tilde{\prob}^{*}$ for the
probability measure $\bprob$ of Example~\ref{Exacorrgbm}, and let
$\tilde{X}$ be the coordinate process taking values in $(0,\infty)^d$.
As shown in Example~\ref{Exacorrgbm}, $X = \exp(a W^{\tilde{\prob}^*})$, where $a$ is the
unique positive definite square root of $A$, and~$W^{\tilde{\prob}^*}$
is a standard Brownian motion under $\tilde{\prob}^*$. Let $\tilde
{Y} =
\tilde{X}/(\mathbf{1}_{d}'\tilde{X})$, where~$\mathbf{1}_d$ is the
vector of all $1$'s in $\reals^d$, and define $Y=(\tilde
{Y}_1,\ldots
,\tilde{Y}_{d-1})$, which is an $E$-valued process. Note that
$\tilde{Y}$ be
recovered from $Y$ since $\tilde{Y}_d = 1 - \sum_{i=1}^{d-1}Y_i$. Using
It\^o's
formula it can be shown that $Y$ has dynamics
\[
dY_t = c(Y_t)\frac{\nabla{\eta}(Y_t)}{\eta(Y_t)}\,\ud t + \tilde{\sigma
}(Y_t)\,\ud W^{\tilde{\prob}^*}_t,
\]
where $\tilde{\sigma}$ is the $(d-1)\times d$ matrix given by
\[
\tilde{\sigma}(x)_{ij} = x_i\Biggl(a_{ij} - \sum_{l=1}^{d-1}x_l
a_{lj} -
\Biggl(1-\sum_{l=1}^{d-1}x_l\Biggr)a_{dj}\Biggr)\qquad \mbox{for } x\in E.
\]
It can be verified that $\tilde{\sigma}\tilde{\sigma}' = c$---indeed,
this is how $c$ was constructed. Thus,
using the one-to-one correspondence between weak solutions of SDEs and
solutions to the Martingale problem (\cite{RW2000}, Chapter 5.4) and the
uniqueness of solutions to the Martingale problem under Assumption \ref
{assbasic} (\cite{MR1326606}, Theorem~1.12.1), it follows that $\prob
^{\eta}[A] = \tilde{\prob}^*[Y\in A]$ holds for all $A\in\F$. Since
$\tilde{X} = \exp(aW^{\tilde{\prob}^*})$, it follows that
$\log
\eta(Y) = \hat{\beta}(*)'a W^{\tilde{\prob}^*} -
\log(\mathbf{1}_d'e^{aW^{\tilde{\prob}^*}})$,
where
\[
\hat{\beta}(*)_i = \hat{\beta}_i,\qquad 1\leq i\leq d-1,\qquad
\hat{\beta}(*)_d = 1-\sum_{j=1}^{d-1}\hat{\beta}_j.
\]
Thus it follows that $\tilde{\prob}^{*}$-a.s., $\lim_{t\uparrow
\infty
}\frac{1}{t}\log\eta(Y_t)= 0$. Hence, with~$X$ denoting the coordinate
process in $E$, $\lim_{t\uparrow\infty}\frac{1}{t}\log\eta(X_t) = 0$
holds $\prob^{\eta}\mbox{-a.s.}$, which implies that $\prob^{\eta
}\in
\bPi_{\mathrm{a.s.}}$. The same
argument as in Example~\ref{Exacorrgbm} yields the
optimality of $\eta$, $\lambda$ and~$\prob^{\eta}$.
\end{exa}

\subsubsection*{An interesting numerical example}

Using the same notation as in Examples~\ref{Exacorrgbm} and
\ref{Exacorrgbmrelcap}, consider for $d=3$ the matrix $A$ and
associated vectors
$\hat{B}$, $\hat{\mathcal{B}}$ given by
\[
A = \pmatrix{ 5/3 & 3 & 0 \cr 3 & 7 & 0 \cr 0 & 0 & 1},\qquad
\hat{B}
=\pmatrix{ -7/4\cr 5/4\cr 1/2}, \qquad\hat{\mathcal{B}} =\pmatrix{
-1\cr 1}.
\]
The eigenvalues of $A$ are $1$ and $13/3(1\pm\sqrt
{145/169}
)$, and hence
$A$ is positive definite. The $\breta$ from (\ref{Egbmbreta}) and
(\ref{Egbmrelcapbreta}), respectively, are
\begin{eqnarray*}
\breta(x,y,x) &=& \sqrt[4]{\frac{y^5z^2}{x^7}}\qquad \mbox{for }
(x,y,z)\in
(0,\infty)^{3},\\
\breta(x,y) &=& \frac{y(1-x-y)}{x}\qquad  \mbox{for } x >0,   y>
0,   x+y<1.
\end{eqnarray*}
Therefore, $\breta$ goes to $\infty$ along the boundary of $E$ in each
case, even when the region is bounded.

\section{\texorpdfstring{Proof of Proposition \protect\ref{Pintprop}}{Proof of Proposition 5.4}} \label{Sonedimproofs}

The proof of Proposition~\ref{Pintprop} relies upon the
following two auxiliary results. As in the proof of Proposition
\ref{Ppwprop}, the symbol $'$ is used to identify derivatives.
%
%
\begin{lem}\label{Luto0}
Let Assumption~\ref{Aonedim} hold. Let $\eta\in
C^2(\alpha,\beta)$ be strictly positive and strictly concave. If
(\ref{Eintblambdazero}) holds, then
\[
\inf_{x\in(\alpha,\beta)}\frac{-c(x)\eta''(x)}{2\eta(x)} = 0.
\]
\end{lem}
\begin{pf}
The proof will be given for the integral near $\alpha$ in
(\ref{Eintblambdazero}); the proof near $\beta$ is the same. Let
$\eta
\in C^2(\alpha,\beta)$ be strictly positive and strictly concave. Set
\[
\delta(\eta) = \inf_{x\in(\alpha,\beta)}\frac{-c(x)\eta
''(x)}{2\eta(x)}.\vadjust{\goodbreak}
\]
Let $x_0\in(\alpha,\beta)$ and normalize $\eta$ so
that $\eta(x_0) = 1$. Note that this will not change the value of
$\delta(\eta)$. Using integration by parts, for $\alpha< x <
x_0$,
\[
\eta(x) = 1 - (x_0 - x)\eta'(x_0) -
\int_x^{x_0}(y-x)(-\eta''(y))\,dy
\]
and hence
\[
\int_\alpha^{x_0} \indic_{\{y\geq x\}}(y-x)(-\eta
''(y))\,dy \leq1 +
(\beta
-\alpha)|\eta'(x_0)|.
\]
Fatou's lemma and the concavity of $\eta$ yield
%
%
\begin{equation}\label{Euintbound}
\int_\alpha^{x_0}(y-\alpha)(-\eta''(y))\,dy \leq1 +
(\beta-\alpha)|\eta'(x_0)|.
\end{equation}
The positivity and concavity of $\eta$ yield for $\alpha< \alpha_m < y
< x_0$ that
\[
\eta(y) = \eta\biggl(\frac{y-\alpha_m}{x_0-\alpha_m}x_0 +
\frac{x_0-y}{x_0-\alpha_m}\alpha_m\biggr)\geq\frac{y-\alpha
_m}{x_0-\alpha_m},
\]
and so, letting $\alpha_m\downarrow\alpha$, it follows that $\eta(y)
\geq
(y-\alpha)/(x_0-\alpha)$.
Thus, if $\delta(\eta) > 0$ and (\ref{Eintblambdazero}) holds, then
\begin{eqnarray*}
\int_\alpha^{x_0}(y-\alpha)(-\eta''(y))\,dy &\geq& 2\delta(\eta)\int
_\alpha
^{x_0}\frac{(y-\alpha)\eta(y)}{c(y)}\,dy\\
&\geq& \frac{2\delta(\eta)}{x_0-\alpha}\int_\alpha^{x_0}\frac
{(y-\alpha
)^2}{c(y)}\,dy   =   \infty,
\end{eqnarray*}
which contradicts (\ref{Euintbound}). Thus, $\delta(\eta) =0$ proving
the result.
\end{pf}
%
%
\begin{lem}\label{Letaposrectest}
Let Assumption~\ref{Aonedim} hold. Let $\lambda> 0$ and $\eta\in
H_{\lambda}$ be
such that
%
%
\begin{equation}\label{Eetaproplim}
\lim_{x\downarrow\alpha}\eta(x) = 0 =\lim_{x\uparrow\beta}\eta(x)
\end{equation}
and
%
%
\begin{equation}\label{Eetapropint}
\int_{\alpha}^{\beta}\frac{\eta^2(x)}{c(x)}\,dx < \infty.
\end{equation}
Then, $\blambda= \lambda$ and $\breta= \eta$. The coordinate process
$X$ under
$(\bprob_x)_{x\in(\alpha,\beta)}$ is positive
recurrent, and hence, by Proposition~\ref{Pblambdatest}, $\breta$ is
unique up to multiplication by a positive constant. Furthermore,
$\bprob
\in\bPi_{\mathrm{a.s.}}$.
\end{lem}
\begin{pf}
If~$X$ is recurrent under $(\bprob_x)_{x\in E}$, then from
Proposition~\ref{Pblambdatest}, $\blambda= \lambda$ and $\breta
=\eta$,
and $\breta$ is unique up to multiplication by a positive constant.
Furthermore, by (\ref{Eetapropint}), positive recurrence will
follow with the invariant measure $\tilde{\eta}$ that has density
proportional to $\eta^2 / c$ with respect to Lebesgue measure,
appropriately normalized so $\tilde{\eta}$ is a probability
measure.\vadjust{\goodbreak}

To check recurrence it will be shown that (\ref{Ebprobrecurrent})
holds near $\alpha$; the proof near $\beta$ is the same. Note that,
since $\eta\in H_{\lambda}$ and
(\ref{Eetaproplim}) holds, there exists a~unique $x_0\in(\alpha
,\beta)$ such
that $\eta'(x_0) = 0$. For $\alpha< x < x_0$,
\[
\int_x^{x_0}\frac{2\lambda\eta(y)^2}{c(y)}\,dy =
-\int_x^{x_0}\eta(y)\eta''(y)\,dy = \eta(x)\eta'(x) +
\int_x^{x_0}\eta'(y)^2\,dy.
\]
Thus, as $x\downarrow\alpha$ since $\eta$ is positive and concave, it
must hold
that $\eta(x)\eta'(x) > 0$, and hence, by (\ref{Eetapropint}), it
follows that $\int_{\alpha}^{x_0}\eta'(y)^2\,dy < \infty$. Therefore, by
the concavity of $\eta$ and (\ref{Eetaproplim}),
%
%
\begin{equation}\label{Eetasquareder}
0\leq\liminf_{x\downarrow\alpha}\eta(x)\eta'(x) \leq
\lim_{x\downarrow\alpha}\int_0^x\eta'(y)^2\,dy = 0.
\end{equation}
This implies that, for any $\eps> 0$, there is an $x_{\eps}$ near
$\alpha$
such that for $x\in(\alpha,x_{\eps})$, $\eta^2(x)\leq2\eps
(x-\alpha)$,
or that
\[
\int_{\alpha}^{x_{\eps}}\frac{1}{\eta(y)^2}\,dy\geq\frac{1}{2\eps
}\int
_{\alpha}^{x_{\eps}}\frac{1}{y-\alpha}\,dy
= \infty,
\]
and recurrence follows. It remains to prove that $\bprob\in\bPi
_{\mathrm{a.s.}}$. To this end, it follows from
equations (\ref{Ebproblogrep}) and (\ref{Eliminfequiv}) in the proof of
Proposition~\ref{Phopfprop} that $\bprob\in\bPi_{\mathrm{a.s.}}$ if
\[
\liminf_{t\uparrow\infty}\frac{1}{t}\int_0^t\biggl(\frac
{1}{2}c(X_s)
\biggl(\frac{\eta'(X_s)}{\eta(X_s)}\biggr)^2-\lambda\biggr)\,ds
\geq0 ,\qquad\bprob\mbox{-a.s.}
\]
By the ergodic theorem (\cite{MR1326606}, Theorem 4.9.5) and the monotone
convergence theorem it follows that
\begin{eqnarray*}
&&\liminf_{t\uparrow\infty}\frac{1}{t}\int_0^t\biggl(\frac
{1}{2}c(X_s)
\biggl(\frac{\eta'(X_s)}{\eta(X_s)}\biggr)^2-\lambda\biggr)\,ds\\
&&\qquad\geq\int
_{\alpha
}^{\beta}\biggl(\frac{1}{2}c(y)\biggl(\frac{\eta'(y)}{\eta
(y)}
\biggr)^2-\lambda\biggr)\frac{\eta(y)^2}{c(y)}\,dy, \qquad
\bprob\mbox{-a.s.}
\end{eqnarray*}
Continuing, $\eta\in H_{\lambda}$ implies
\[
\int_{\alpha}^{\beta}\biggl(\frac{1}{2}c(y)\biggl(\frac{\eta
'(y)}{\eta
(y)}\biggr)^2-\lambda\biggr)\frac{\eta(y)^2}{c(y)}\,dy
= \lim_{x\downarrow\alpha}\eta(x)\eta'(x) -
\lim_{x\uparrow\beta}\eta(x)\eta'(x) =0,
\]
where the last equality follows from (\ref{Eetasquareder}) since the same
equality holds near~$\beta$. Thus, $\bprob\in\bPi_{\mathrm{a.s.}}$.
\end{pf}

In what follows, the proof of Proposition~\ref{Pintprop} will be given.

The proof of how (\ref{Eintblambdazero}) implies $\blambda
= 0$ is handled first. By (\ref{Eblambdapw01rep}), it suffices to
consider strictly concave functions $\eta$. However, since
(\ref{Eintblambdazero}) holds, Lemma~\ref{Luto0} applies and hence
$\delta(\eta) = 0$ for all such $\eta$. Thus $\blambda= 0$.

Regarding the assertions when (\ref{Eintblambdapossl}) holds, in light
of Lemma
\ref{Letaposrectest} it suffices to show that (\ref{Eintblambdapossl})
yields the existence of a $\lambda> 0$, $\eta\in H_{\lambda}$ such that\vadjust{\goodbreak}
conditions (\ref{Eetaproplim}) and (\ref{Eetapropint}) are satisfied. To
this end, define the $\sigma$-finite measure $m$ via $m(dx) =
c(x)^{-1}\,dx$. Note that condition (\ref{Eetapropint}) now reads $\eta
\in
L^2((\alpha,\beta),m)$. The desired pair $(\lambda,\eta)$ are the principle
eigenvalue and eigenfunction for the operator $(L,\mathcal{D}(L))$
where $(L\eta)(x) = -(1/2)c(x)\eta''(x)$ for $x \in(\alpha, \beta)$,
and the domain $\mathcal{D}(L)$ consists of functions which vanish at
$\alpha,\beta$ and is constructed so that $(L,\mathcal{D}(L))$ is
self adjoint
in $L^2((\alpha,\beta),m)$. $\mathcal{D}(L)$ is highly dependent
upon the
behavior of $m$ near $\alpha$ and $\beta$. The study of the spectral
properties of such operators falls under the name
\textit{Sturm--Liouville theory}. For a detailed exposition on the topics
covered/results given
below, see~\cite{MR1784426} and~\cite{MR2170950}.

The case when $m((\alpha,\beta)) < \infty$ is called the
\textit{regular case}. Here $\mathcal{D}(L)$ is
given by
%
%
\begin{eqnarray}\label{ELdomaindef}
\mathcal{D}(L) &=& \{\eta\in L^2((\alpha,\beta),m) \such\eta
'\in
AC(\alpha,\beta), \eta(\alpha) = \eta(\beta) = 0,\nonumber\\[-8pt]\\[-8pt]
&&\qquad\hspace*{126pt} c\eta''\in
L^2((\alpha
,\beta),m)\},\nonumber
\end{eqnarray}
and the existence of a $\lambda> 0$, $\eta\in H_{\lambda}\cap
\mathcal
{D}(L)$ is given by
\cite{MR1784426}, Theorem 2.7.4, and~\cite{MR2170950}, Theorem 10.12.1.

Now, suppose that (\ref{Eintblambdapossl}) holds, but for some $a\in
(\alpha,\beta)$ either $m((\alpha,a))=\infty$ or $m((a,\beta)) =
\infty$, or both. These cases are called the \textit{singular cases}. In
each of these three cases there exists a domain $\mathcal{D}(L)\subset
L^2((\alpha,\beta),m)$, similar to that in~(\ref{ELdomaindef}), such
that $(L,\mathcal{D}(L))$ is self adjoint. For explicit formulas for
the domains, see~\cite{MR2170950}, Chapters 7 and 10.

According to~\cite{MR2170950}, Theorem 10.12.1(8), if the spectrum of
$(L,\mathcal{D}(L))$ is discrete and bounded from below, then in fact there
exists a $\lambda> 0$ and $\eta\in H_{\lambda}\cap\mathcal{D}(L)$
such that (\ref{Eetaproplim}) holds [this last fact follows by construction of
$\mathcal{D}(L)$ but also because otherwise $\eta\notin L^2((\alpha
,\beta),m)$].

To prove the spectrum is discrete and bounded from
below, it suffices to treat the case of one regular and one singular endpoint.
This follows using the \textit{spectral decomposition method} on which a
detailed description may be found in~\cite{MR0190800}. Without loss of
generality, consider the case when $\alpha$ is regular and $\beta$ is
singular. Under the transformation $z = f(x) = \int_{\alpha}^{x}
(1/ c(y))\,dy$, $(\alpha,\beta)$ is taken to be $(0,\infty
)$. Set
$\varphi
(z) = \eta(x)$ and
$g(z) = f^{-1}(z)$. Note that $\eta\in
L^2((\alpha,\beta),m)$ is equivalent to $\varphi\in
L^2((0,\infty),\mathrm{Leb}) \equiv L^2(0,\infty)$. Furthermore, the
operator $(M,\mathcal{D}(M))$ defined by
\[
(M\varphi) (z) =
-\frac{1}{2}\biggl(\frac{1}{g'(z)}\varphi'(z)\biggr)',\qquad
\mathcal{D}(M) = \{\varphi\such
\varphi(z)=\eta(x), \eta\in\mathcal{D}(L)\}
\]
is self-adjoint in $L^2(0,\infty)$. Let $N>0$ and
\[
Q_N = \{v\in C_0((N,\infty),\mathbb{C}) \such v\in
AC_{\mathrm{loc}}(0,\infty),
v'\in L^2(0,\infty)\},
\]
where $C_0$ means that $v$ is continuous and compactly supported in
$(N,\infty)$. For $v\in Q_N$, set
\[
I(v,N) = \frac{1}{2}\int_N^{\infty}
\frac{|v'(z)|^2}{g'(z)}\,dz.
\]

According to~\cite{MR614218}, Lemma 4.2, $(M,\mathcal{D}(M))$ has a discrete
spectrum bound\-ed from below if and only if for each $\theta>0$ there
exists an $N>0$ such that
\[
I(v,N)\geq\theta\int_{N}^{\infty}v(z)^2\,dz
\]
for each real valued $v\in Q_N$. To show this, fix $\theta
>0$. For any $N>0$ and $v\in Q_N$,
\[
v(z) = -\int_z^{\infty} v'(\tau)\,d\tau.
\]
Since $\tau= f(g(\tau))$, it follows that $g'(\tau) = c(g(\tau)) > 0$.
By H\"{o}lder's inequality, for real valued $v\in Q_N$,
\[
v(z)^2
\leq\biggl(\int_z^{\infty}\frac{v'(\tau)^2}{g'(\tau)}\,d\tau
\biggr)
\biggl(\int_z^{\infty}g'(\tau)\,d\tau\biggr)\leq
2I(v,N)\bigl(\beta- g(z)\bigr).
\]
Therefore,
\begin{eqnarray*}
\theta\int_N^{\infty}v(z)^2\,dz &\leq& 2\theta I(v,N)\int_N^{\infty
}\bigl(\beta
-g(z)\bigr)\,dz\\
&=& 2\theta I(v,N)\int_{g(N)}^{\beta}\frac{\beta-x}{c(x)}\,dx,
\end{eqnarray*}
where the last equality follows from the substitution $x = g(z)$ or $z =
f(x)$. Since $\lim_{z\uparrow\infty}g(x) = \beta$, by
(\ref{Eintblambdapossl})
\[
2\theta\int_{g(N)}^\beta\frac{\beta-x}{c(x)}\,dx\leq1
\]
for $N$ large enough, yielding the desired result. \qed

\section*{Acknowledgments}

The authors would like to express their gratitude toward the Associate
Editor and four anonymous referees for their very constructive remarks
that helped improve the paper significantly.


%

\printaddresses


\begin{thebibliography}{31}

\bibitem{MR1159579}
%
\begin{barticle}[mr]
\bauthor{\bsnm{Algoet},~\bfnm{Paul}\binits{P.}}
(\byear{1992}).
\btitle{Universal schemes for prediction, gambling and portfolio selection}.
\bjournal{Ann. Probab.}
\bvolume{20}
\bpages{901--941}.
\bid{issn={0091-1798}, mr={1159579}}
\bptok{imsref}%
\end{barticle}
%
\endbibitem

\bibitem{MR929084}
%
\begin{barticle}[mr]
\bauthor{\bsnm{Algoet},~\bfnm{Paul~H.}\binits{P.~H.}} \AND
\bauthor{\bsnm{Cover},~\bfnm{Thomas~M.}\binits{T.~M.}}
(\byear{1988}).
\btitle{Asymptotic optimality and asymptotic equipartition properties of
log-optimum investment}.
\bjournal{Ann. Probab.}
\bvolume{16}
\bpages{876--898}.
\bid{issn={0091-1798}, mr={0929084}}
\bptok{imsref}%
\end{barticle}
%
\endbibitem

\bibitem{BayKarYao10}
%
\begin{bmisc}[auto:STB|2012/01/09|08:49:38]
\bauthor{\bsnm{Bayraktar},~\bfnm{E.}\binits{E.}},
\bauthor{\bsnm{Karatzas},~\bfnm{I.}\binits{I.}} \AND
\bauthor{\bsnm{Yao},~\bfnm{S.}\binits{S.}}
(\byear{2012}).
\bhowpublished{Optimal stopping for dynamic convex risk measures.
\textit{Illinois J. Math.} To appear.}
\bptok{imsref}%
\end{bmisc}
%
\endbibitem

\bibitem{MR1478722}
%
\begin{bincollection}[mr]
\bauthor{\bsnm{Elworthy},~\bfnm{K.~D.}\binits{K.~D.}},
\bauthor{\bsnm{Li},~\bfnm{X.~M.}\binits{X.~M.}} \AND
\bauthor{\bsnm{Yor},~\bfnm{M.}\binits{M.}}
(\byear{1997}).
\btitle{On the tails of the supremum and the quadratic variation of strictly
local martingales}.
In \bbooktitle{S\'eminaire de {P}robabilit\'es, {XXXI}}.
\bseries{Lecture Notes in Math.}
\bvolume{1655}
\bpages{113--125}.
\bpublisher{Springer}, \baddress{Berlin}.
\bid{doi={10.1007/BFb0119298}, mr={1478722}}
\bptok{imsref}%
\end{bincollection}
%
\endbibitem

\bibitem{MR1625845}
%
\begin{bbook}[mr]
\bauthor{\bsnm{Evans},~\bfnm{Lawrence~C.}\binits{L.~C.}}
(\byear{1998}).
\btitle{Partial Differential Equations}.
\bseries{Graduate Studies in Mathematics}
\bvolume{19}.
\bpublisher{Amer. Math. Soc.}, \baddress{Providence, RI}.
\bid{mr={1625845}}
\bptok{imsref}%
\end{bbook}
%
\endbibitem

\bibitem{KarFern10a}
%
\begin{barticle}[mr]
\bauthor{\bsnm{Fernholz},~\bfnm{Daniel}\binits{D.}} \AND
\bauthor{\bsnm{Karatzas},~\bfnm{Ioannis}\binits{I.}}
(\byear{2010}).
\btitle{On optimal arbitrage}.
\bjournal{Ann. Appl. Probab.}
\bvolume{20}
\bpages{1179--1204}.
\bid{doi={10.1214/09-AAP642}, issn={1050-5164}, mr={2676936}}
\bptok{imsref}%
\end{barticle}
%
\endbibitem

\bibitem{KarFern10c}
%
\begin{barticle}[mr]
\bauthor{\bsnm{Fernholz},~\bfnm{Daniel}\binits{D.}} \AND
\bauthor{\bsnm{Karatzas},~\bfnm{Ioannis}\binits{I.}}
(\byear{2011}).
\btitle{Optimal arbitrage under model uncertainty}.
\bjournal{Ann. Appl. Probab.}
\bvolume{21}
\bpages{2191--2225}.
\bptok{imsref}%
\end{barticle}
%
\endbibitem

\bibitem{KarFern10b}
%
\begin{bincollection}[mr]
\bauthor{\bsnm{Fernholz},~\bfnm{Daniel}\binits{D.}} \AND
\bauthor{\bsnm{Karatzas},~\bfnm{Ioannis}\binits{I.}}
(\byear{2010}).
\btitle{Probabilistic aspects of arbitrage}.
In \bbooktitle{Contemporary Quantitative Finance}
\bpages{1--17}.
\bpublisher{Springer}, \baddress{Berlin}.
\bid{mr={2732837}}
\bptok{imsref}%
\end{bincollection}
%
\endbibitem

\bibitem{MR1861997}
%
\begin{barticle}[mr]
\bauthor{\bsnm{Fernholz},~\bfnm{Robert}\binits{R.}}
(\byear{2001}).
\btitle{Equity portfolios generated by functions of ranked market weights}.
\bjournal{Finance Stoch.}
\bvolume{5}
\bpages{469--486}.
\bid{doi={10.1007/s007800100044}, issn={0949-2984}, mr={1861997}}
\bptok{imsref}%
\end{barticle}
%
\endbibitem

\bibitem{MR2247836}
%
\begin{barticle}[mr]
\bauthor{\bsnm{F{\"o}llmer},~\bfnm{Hans}\binits{H.}} \AND
\bauthor{\bsnm{Gundel},~\bfnm{Anne}\binits{A.}}
(\byear{2006}).
\btitle{Robust projections in the class of martingale measures}.
\bjournal{Illinois J. Math.}
\bvolume{50}
\bpages{439--472 (electronic)}.
\bid{issn={0019-2082}, mr={2247836}}
\bptok{imsref}%
\end{barticle}
%
\endbibitem

\bibitem{gilboa1989meu}
%
\begin{barticle}[mr]
\bauthor{\bsnm{Gilboa},~\bfnm{Itzhak}\binits{I.}} \AND
\bauthor{\bsnm{Schmeidler},~\bfnm{David}\binits{D.}}
(\byear{1989}).
\btitle{Maxmin expected utility with nonunique prior}.
\bjournal{J.~Math. Econom.}
\bvolume{18}
\bpages{141--153}.
\bid{doi={10.1016/0304-4068(89)90018-9}, issn={0304-4068}, mr={1000102}}
\bptok{imsref}%
\end{barticle}
%
\endbibitem

\bibitem{MR0190800}
%
\begin{bbook}[mr]
\bauthor{\bsnm{Glazman},~\bfnm{I.~M.}\binits{I.~M.}}
(\byear{1966}).
\btitle{Direct Methods of Qualitative Spectral Analysis of Singular
Differential Operators}.
\bpublisher{Israel Program for Scientific Translations},
\baddress{Jerusalem}.
\bnote{Translated from the Russian by the IPST Staff}.
\bid{mr={0190800}}
\bptok{imsref}%
\end{bbook}
%
\endbibitem

\bibitem{MR2211122}
%
\begin{barticle}[mr]
\bauthor{\bsnm{Gundel},~\bfnm{Anne}\binits{A.}}
(\byear{2005}).
\btitle{Robust utility maximization for complete and incomplete market models}.
\bjournal{Finance Stoch.}
\bvolume{9}
\bpages{151--176}.
\bid{doi={10.1007/s00780-004-0148-1}, issn={0949-2984}, mr={2211122}}
\bptok{imsref}%
\end{barticle}
%
\endbibitem

\bibitem{MR2319423}
%
\begin{barticle}[mr]
\bauthor{\bsnm{Gy{\"o}rfi},~\bfnm{L{\'a}szl{\'o}}\binits{L.}},
\bauthor{\bsnm{Urb{\'a}n},~\bfnm{Andr{\'a}s}\binits{A.}} \AND
\bauthor{\bsnm{Vajda},~\bfnm{Istv{\'a}n}\binits{I.}}
(\byear{2007}).
\btitle{Kernel-based semi-log-optimal empirical portfolio selection
strategies}.
\bjournal{Int. J. Theor. Appl. Finance}
\bvolume{10}
\bpages{505--516}.
\bid{doi={10.1142/S0219024907004251}, issn={0219-0249}, mr={2319423}}
\bptok{imsref}%
\end{barticle}
%
\endbibitem

\bibitem{MR1121940}
%
\begin{bbook}[mr]
\bauthor{\bsnm{Karatzas},~\bfnm{Ioannis}\binits{I.}} \AND
\bauthor{\bsnm{Shreve},~\bfnm{Steven~E.}\binits{S.~E.}}
(\byear{1991}).
\btitle{Brownian Motion and Stochastic Calculus},
\bedition{2nd} ed.
\bseries{Graduate Texts in Mathematics}
\bvolume{113}.
\bpublisher{Springer}, \baddress{New York}.
\bid{mr={1121940}}
\bptok{imsref}%
\end{bbook}
%
\endbibitem

\bibitem{Kard10}
%
\begin{barticle}[mr]
\bauthor{\bsnm{Kardaras},~\bfnm{Constantinos}\binits{C.}}
(\byear{2010}).
\btitle{The continuous behavior of the num\'eraire portfolio under small
changes in information structure, probabilistic views and investment
constraints}.
\bjournal{Stochastic Process. Appl.}
\bvolume{120}
\bpages{331--347}.
\bid{doi={10.1016/j.spa.2009.11.006}, issn={0304-4149}, mr={2584897}}
\bptok{imsref}%
\end{barticle}
%
\endbibitem

\bibitem{Knipsel2}
%
\begin{bmisc}[auto:STB|2012/01/09|08:49:38]
\bauthor{\bsnm{Knispel},~\bfnm{T.}\binits{T.}}
(\byear{2010}).
\bhowpublished{Asymptotic minimization of robust ``downside risk.'' Available
at
\href
{http://www.stochastik.uni-hannover.de/fileadmin/institut/pdf/DownsideR%
isk.pdf}{www. stochastik.uni-hannover.de/fileadmin/institut/pdf/DownsideRisk.pdf}.}
\bptok{imsref}%
\end{bmisc}
%
\endbibitem

\bibitem{Knipsel1}
%
\begin{bmisc}[auto:STB|2012/01/09|08:49:38]
\bauthor{\bsnm{Knispel},~\bfnm{T.}\binits{T.}}
(\byear{2010}).
\bhowpublished{Asymptotics of robust utility maximization. Available at
\href
{http://www.stochastik.uni-hannover.de/fileadmin/institut/pdf/PowerUtil%
ity.pdf}{www.stochastik. uni-hannover.de/fileadmin/institut/pdf/PowerUtility.pdf}.}
\bptok{imsref}%
\end{bmisc}
%
\endbibitem

\bibitem{MR614218}
%
\begin{barticle}[mr]
\bauthor{\bsnm{Kwong},~\bfnm{Man~Kam}\binits{M.~K.}} \AND
\bauthor{\bsnm{Zettl},~\bfnm{A.}\binits{A.}}
(\byear{1981}).
\btitle{Discreteness conditions for the spectrum of ordinary differential
operators}.
\bjournal{J. Differential Equations}
\bvolume{40}
\bpages{53--70}.
\bid{doi={10.1016/0022-0396(81)90010-3}, issn={0022-0396}, mr={0614218}}
\bptok{imsref}%
\end{barticle}
%
\endbibitem

\bibitem{MR1784426}
%
\begin{bbook}[mr]
\bauthor{\bsnm{Miklav{\v{c}}i{\v{c}}},~\bfnm{Milan}\binits{M.}}
(\byear{1998}).
\btitle{Applied Functional Analysis and Partial Differential Equations}.
\bpublisher{World Scientific}, \baddress{River Edge, NJ}.
\bid{mr={1784426}}
\bptok{imsref}%
\end{bbook}
%
\endbibitem

\bibitem{OstRhe}
%
\begin{barticle}[auto:STB|2012/01/09|08:49:38]
\bauthor{\bsnm{Osterrieder},~\bfnm{J.}\binits{J.}} \AND
\bauthor{\bsnm{Rheinl{\"a}nder},~\bfnm{T.}\binits{T.}}
(\byear{2006}).
\btitle{Arbitrage opportunities in diverse markets via a non-equivalent
measure change}.
\bjournal{Annals of Finance}
\bvolume{2}
\bpages{287--301}.
\bptok{imsref}%
\end{barticle}
%
\endbibitem

\bibitem{MR1152459}
%
\begin{barticle}[mr]
\bauthor{\bsnm{Pinchover},~\bfnm{Yehuda}\binits{Y.}}
(\byear{1992}).
\btitle{Large time behavior of the heat kernel and the behavior of the {G}reen
function near criticality for nonsymmetric elliptic operators}.
\bjournal{J. Funct. Anal.}
\bvolume{104}
\bpages{54--70}.
\bid{doi={10.1016/0022-1236(92)90090-6}, issn={0022-1236}, mr={1152459}}
\bptok{imsref}%
\end{barticle}
%
\endbibitem

\bibitem{Pinchover95}
%
\begin{barticle}[mr]
\bauthor{\bsnm{Pinchover},~\bfnm{Yehuda}\binits{Y.}}
(\byear{1995}).
\btitle{On nonexistence of any {$\lambda\sb0$}-invariant positive harmonic
function, a counter example to {S}troock's conjecture}.
\bjournal{Comm. Partial Differential Equations}
\bvolume{20}
\bpages{1831--1846}.
\bid{doi={10.1080/03605309508821153}, issn={0360-5302}, mr={1349233}}
\bptok{imsref}%
\end{barticle}
%
\endbibitem

\bibitem{MR781410}
%
\begin{barticle}[mr]
\bauthor{\bsnm{Pinsky},~\bfnm{Ross~G.}\binits{R.~G.}}
(\byear{1985}).
\btitle{On the convergence of diffusion processes conditioned to remain
in a
bounded region for large time to limiting positive recurrent diffusion
processes}.
\bjournal{Ann. Probab.}
\bvolume{13}
\bpages{363--378}.
\bid{issn={0091-1798}, mr={0781410}}
\bptok{imsref}%
\end{barticle}
%
\endbibitem

\bibitem{MR1326606}
%
\begin{bbook}[mr]
\bauthor{\bsnm{Pinsky},~\bfnm{Ross~G.}\binits{R.~G.}}
(\byear{1995}).
\btitle{Positive Harmonic Functions and Diffusion}.
\bseries{Cambridge Studies in Advanced Mathematics}
\bvolume{45}.
\bpublisher{Cambridge Univ. Press}, \baddress{Cambridge}.
\bid{doi={10.1017/CBO9780511526244}, mr={1326606}}
\bptok{imsref}%
\end{bbook}
%
\endbibitem

\bibitem{MR2096294}
%
\begin{bincollection}[mr]
\bauthor{\bsnm{Quenez},~\bfnm{Marie-Claire}\binits{M.-C.}}
(\byear{2004}).
\btitle{Optimal portfolio in a multiple-priors model}.
In \bbooktitle{Seminar on {S}tochastic {A}nalysis, {R}andom {F}ields and
{A}pplications {IV}}.
\bseries{Progress in Probability}
\bvolume{58}
\bpages{291--321}.
\bpublisher{Birkh\"auser}, \baddress{Basel}.
\bid{mr={2096294}}
\bptok{imsref}%
\end{bincollection}
%
\endbibitem

\bibitem{RW2000}
%
\begin{bbook}[mr]
\bauthor{\bsnm{Rogers},~\bfnm{L.~C.~G.}\binits{L.~C.~G.}} \AND
\bauthor{\bsnm{Williams},~\bfnm{David}\binits{D.}}
(\byear{2000}).
\btitle{Diffusions, {M}arkov Processes, and Martingales}.
\bseries{Cambridge Mathematical Library}
\bvolume{2}.
\bpublisher{Cambridge Univ. Press}, \baddress{Cambridge}.
\bid{mr={1780932}}
\bptok{imsref}%
\end{bbook}
%
\endbibitem

\bibitem{MR2284014}
%
\begin{barticle}[mr]
\bauthor{\bsnm{Schied},~\bfnm{Alexander}\binits{A.}}
(\byear{2007}).
\btitle{Optimal investments for risk- and ambiguity-averse
preferences: A
duality approach}.
\bjournal{Finance Stoch.}
\bvolume{11}
\bpages{107--129}.
\bid{doi={10.1007/s00780-006-0024-2}, issn={0949-2984}, mr={2284014}}
\bptok{imsref}%
\end{barticle}
%
\endbibitem

\bibitem{MR2236457}
%
\begin{barticle}[mr]
\bauthor{\bsnm{Schied},~\bfnm{Alexander}\binits{A.}} \AND
\bauthor{\bsnm{Wu},~\bfnm{Ching-Tang}\binits{C.-T.}}
(\byear{2005}).
\btitle{Duality theory for optimal investments under model uncertainty}.
\bjournal{Statist. Decisions}
\bvolume{23}
\bpages{199--217}.
\bid{doi={10.1524/stnd.2005.23.3.199}, issn={0721-2631}, mr={2236457}}
\bptok{imsref}%
\end{barticle}
%
\endbibitem

\bibitem{MR2170950}
%
\begin{bbook}[mr]
\bauthor{\bsnm{Zettl},~\bfnm{Anton}\binits{A.}}
(\byear{2005}).
\btitle{Sturm--{L}iouville Theory}.
\bseries{Mathematical Surveys and Monographs}
\bvolume{121}.
\bpublisher{Amer. Math. Soc.}, \baddress{Providence, RI}.
\bid{mr={2170950}}
\bptok{imsref}%
\end{bbook}
%
\endbibitem

\end{thebibliography}
\end{document}